%% file: main.tex
\documentclass{article}
\usepackage[dvipsnames]{xcolor}
\usepackage{graphicx} % Required for inserting images
\usepackage{times}
\usepackage{soul}
\usepackage{url}
\usepackage[hidelinks]{hyperref}
\usepackage[utf8]{inputenc}
\usepackage[small]{caption}
\usepackage[margin=1in]{geometry}
\usepackage{amsmath}
\usepackage{amsthm}
\usepackage{booktabs}
\usepackage{algorithm}
\usepackage{algorithmic}
\usepackage{lineno}

\usepackage{amssymb}
\usepackage{tikz}
\usetikzlibrary{patterns}
\usepackage{subcaption}
\usepackage{bm}

\newtheorem{theorem}{Theorem}
\newtheorem{definition}{Definition}
\newtheorem{lemma}{Lemma}

\title{Finding Possible Winners in Spatial Voting with Incomplete Information}

\author{
Hadas Shachnai\footnote{Computer Science Department, Technion, Haifa, Israel. \texttt{hadas@cs.technion.ac.il}}
\and
Rotem Shavitt\footnote{Computer Science Department, Technion, Haifa, Israel. \texttt{rshavitt@gmail.com}. Corresponding author.}
\and
Andreas Wiese\footnote{
Mathematics Department, Technical University of Munich, Germany. \texttt{andreas.wiese@tum.de}}
}

\date{}

\begin{document}

\global\long\def\H{\mathcal{F}}%
\global\long\def\shape{f}%
\global\long\def\N{\mathbb{N}}%
\global\long\def\R{\mathbb{R}}%
\global\long\def\T{\mathrm{T}}%
\newcommand{\mL}{\reflectbox{L}}
\newcommand{\NmL}{\textnormal{\reflectbox{L}}}
\newcommand{\cS}{\mathcal{S}}
\newcommand{\s}{\mathbf {s}}
\newcommand{\M}{\cal {M}}

\def\nw{\mathsf{NW}}
\def\pw{\mathsf{PW}}
\def\wpw{\mathsf{WPW}}
\def\angs#1{\mathord{\langle{#1\rangle}}}
\def\nwpar#1{\mathsf{NW}\angs{#1}}
\def\pwpar#1{\mathsf{PW}\angs{#1}}
\def\wpwpar#1{\mathsf{WPW}\angs{#1}}

\maketitle

\begin{abstract}

\input{abstract}  

\end{abstract}

\section{Introduction}\label{sec:intro}

\input{intro}

\subsection{Related Work}
\label{sec:related}

\input{related_work}

\section{Preliminaries}\label{sec:prel}

\input{prelims}

\section{$\pwpar{1}$ with $k$-Truncated Voting Rules}\label{sec:truncated}

\input{k-truncated}

\input{fpt}

\section{Spatial Voting with Weighted Voters}\label{sec:weighted}

\input{weighted}

\section{Conclusion}
\label{sec:discussion}

\input{discussion}

\bibliographystyle{abbrv}
\bibliography{resource}

\end{document}

%% file: abstract.tex
We consider a spatial voting model where both candidates and voters are positioned in the $d$-dimensional Euclidean space, and each voter ranks candidates based on their proximity to the voter's ideal point.
We focus on the scenario where the given information about the locations of the voters' ideal points is incomplete; for each dimension, only an interval of possible values is known. In this context, we investigate the computational complexity of determining the possible winners under positional scoring rules.
Our results show that the possible winner problem in one dimension is solvable in polynomial time for all $k$-truncated voting rules with constant $k$. Moreover, for some scoring rules for which the possible winner problem is NP-complete, such as approval voting for any dimension or $k$-approval for $d \geq 2$ dimensions, we give an FPT algorithm parameterized by the number of candidates. Finally, we classify tractable and intractable settings of the {\em weighted} possible winner problem in one dimension, and resolve the computational complexity of the weighted case for all two-valued positional scoring rules when $d=1$.

%% file: intro.tex
The spatial model of voting associates voters and candidates with points in the $d$-dimensional Euclidean space, i.e., in~$\mathbb{R}^d$. Each dimension corresponds to an issue on which the voters and the candidates have an opinion; this opinion is defined by the coordinate of the voter or candidate in this dimension. Voters prefer candidates closer to their respective point (measured as the Euclidean distance in $\mathbb{R}^d$) over those who are further away. Hence, for each voter this induces an order of the candidates. %, describing the preferences of this voter.
In the social choice literature, preferences with this structure are often referred to as ($d$-)Euclidean preferences~\cite{bogomolnaia2007euclidean,elkind2022preference}. The most common example of a spatial model is a political spectrum, such as the traditional left-right axis where $d=1$, but issue spaces can be of higher dimension (see, e.g., \cite{AGG2015}).

We consider a common scenario where the point in $\mathbb{R}^d$ of each candidate is known precisely, e.g., from the election campaign,
but for the voters' preferences only partial information is available. For each voter and each of the $d$ dimensions, we assume that we are given an interval which contains the opinion of the voter corresponding to this dimension.
This model captures real-world uncertainty in political elections, where it is often difficult to determine exactly which party a voter supports. However, we can typically estimate a range for their views—for instance, whether they tend to lean left, right, or center.
From this partial information we can identify a set of possible preference orders for the voter.

We study voting systems in which there is a global scoring vector $\vec{s}_m=(s_m(1),s_m(2),...,s_m(m))$, depending on the number of candidates, with $s_m(1)\ge s_m(2)\ge ...\ge s_m(m)$ such that each voter gives $s_m(1)$ votes to her favorite candidate, $s_m(2)$ votes to her second favorite candidate, and so on.
Also, we study \emph{approval voting} where each voter $v_j$ gives one vote to each candidate whose opinion is within a given \emph{approval radius} $\rho_j$ of the point in $\mathbb{R}^d$ corresponding to $v_j$. In both settings, we say that a candidate \emph{can win the election} if no other candidate receives a higher total score.

Since the precise opinion of each voter is not known, it is unclear which candidate will win the election. Two key questions arising in this setting
are whether a specific candidate can be a {\em possible winner} (who wins in at least one scenario by the opinions of the voters) or a {\em necessary winner} (one who wins in every possible scenario). These questions, introduced in the seminal work of~\cite{konczak2005voting}, 
have garnered significant attention in various settings involving incomplete information about voters' preferences (see Section~\ref{sec:related}).

\input{results_table}

The {\em necessary winner} problem ($\nw$) is well understood in this model, thanks  to a thorough study in Imber et al.~\cite{imber2024spatial}. However,
the complexity of $\pwpar{d}$, i.e., the \emph{possible winner} problem with incomplete voters' information in $d$ dimensions,
is known only for certain classes of scoring rules (defined via certain classes of vectors~$\s$).
Specifically, as shown in~\cite{imber2024spatial},
$\pwpar{1}$ is solvable in polynomial time for all two-valued rules, i.e., rules in which the vector $\vec{s}$ contains only two different values, and for two specific families of rules with more than two values: $(i)$ The three-valued rule $F(k,t)$, in which the scoring vector is $s=(2, \ldots , 2,1, \ldots, 1, 0, \ldots , 0)$, starting with $k$ occurrences of two and ending with $t$ zeroes; $\pwpar{1}$ is tractable for $F(k,t)$ whenever $k > t$. $(ii)$ Weighted veto rules, which are of the form  $s=(\alpha, \ldots, \ldots, \beta_1, \ldots, \beta_k)$ for $\alpha > \beta_1 \geq \cdot \cdot \cdot \geq \beta_k$ with $k < m/2$.
On the other hand, $\pwpar{d}$ is NP-complete for any number of dimensions $d\ge 2$, already for the (relatively simple) scoring vector $\s=(1,1,1,0,...,0)$~\cite{imber2024spatial}. 

These previous results leave several intriguing questions open: (1) Is $\pwpar{1}$ still tractable for other positional scoring rules with more than two values, e.g., for $\vec{s}=(2,1,0,...,0)$?
(2) For voting rules under which $\pwpar{d}$ is NP-complete, can we devise parameterized algorithms? (3) What happens if each voter is associated with a weight, e.g., representing a group of voters sharing the same (unknown) common opinion, or members with varying levels of influence in a board of directors? Is the weighted possible winner problem $\wpwpar{d}$ harder than $\pwpar{d}$?
We answer all three questions positively, see Table~\ref{tab:results} for an overview.

First, we present a polynomial-time algorithm for $\pwpar{1}$ for any scoring vector $\vec{s}$ with a constant number of non-zero entries.
Such scoring vectors are very common in real-life voting systems: the Eurovision Song Contest~\cite{SBKC18}) uses the scoring vector $\vec{s}=(12,10,8,7,6,5,4,3,2,1,0...,0)$ and the NBA MVP contest uses 
$\vec{s}=(10,7,5,3,1,0,\dots, 0)$. 
Also, this class contains the $k$-truncated Borda rule $(k,k-1,...,1,0,\ldots,0)$ which is used in the NCAA Football Division 1A Coaches’ poll for $k=25$.

Such scoring rules are popular because ranking all candidates becomes impractical when there are many candidates. Moreover, voters often lack strong preferences beyond their top choices, making the order among lower-ranked candidates irrelevant. 
Also, keeping the number of positive entries fixed ensures stability in the voting system in repeated contests as described above, where the number of candidates may vary, and a voting rule that depends on this number complicates the process and hinders comparisons across events.

Our algorithm reduces $\pw$ to the problem of {\em shapes scheduling}. To the best of our knowledge, this scheduling setting has not been studied before and it might be of independent interest. We solve the resulting instances of this problem, building on a technique of Baptiste~\cite{baptiste2000scheduling}.

In real elections, the number of candidates $m$ (with realistic chances of winning) is typically rather small. This motivates us to choose $m$ as a fixed parameter. We show that for any dimension $d$ (not necessarily constant or bounded by a fixed parameter) the problem $\pwpar{d}$ becomes fixed-parameter tractable (FPT) for \emph{any} scoring vector $s$ and also for approval voting, i.e., we can solve the problem in a running time of the form $f(m)\cdot n^{O(1)}$ for some computable function $f$.

Finally, we prove that $\wpwpar{d}$ is NP-complete already when $d=1$ and $m=4$ under the Borda scoring rule. In contrast, our result above shows that the corresponding unweighted case admits a polynomial time algorithm. 
In addition, we resolve the computational complexity of the weighted possible winner problem for all two-valued positional scoring rules when $d=1$, by distinguishing between voting rules which remain tractable, and others under which $\wpwpar{1}$ becomes NP-complete (for short, NP-c).

%% file: results_table.tex
\begin{table*}[ht]
\centering
\begin{tabular}{@{}l|ccc@{}}
\toprule
\textbf{Problem} & \textbf{$k$-approval} & \textbf{Multi-valued positional scoring rules} & \textbf{Approval voting} \\ 
\midrule
$\pwpar{1}$ & in P \cite{imber2024spatial} & \begin{tabular}[c]{@{}c@{}} in P for any $k$-truncated scoring\\ rule for a constant $k$ [Theorem \ref{thm:k-truncated}] \end{tabular} &
\begin{tabular}[c]{@{}c@{}}NP-c \cite{imber2024spatial} \\ FPT in $m$ [Theorem~\ref{theorem:approval_fpt}] \end{tabular} \\[1em]
$\pwpar{d}$ & 
\begin{tabular}[c]{@{}c@{}}NP-c for $k\geq3$, $d\geq 2$ \cite{imber2024spatial}, \\ FPT in $m$ [Theorem~\ref{theorem:approval_positional}] \end{tabular} & 
FPT in $m$ [Theorem \ref{theorem:approval_positional}] & \begin{tabular}[c]{@{}c@{}}NP-c \cite{imber2024spatial} \\ FPT in $m$ [Theorem~\ref{theorem:approval_fpt}] \end{tabular} \\[1em]
$\wpwpar{1}$  &  
\begin{tabular}[c]{@{}c@{}}in P if  $k(m)\geq \frac{m}{2}$ $\forall m\in \N$, \\ otherwise NP-c [Theorem \ref{theorem:k(m)_npc}]\end{tabular} &
\begin{tabular}[c]{@{}c@{}}NP-c for Borda with $m\geq 4$ \\[0em]  [Theorem \ref{theorem:weighted_borda}] \end{tabular} & 
NP-c \cite{imber2024spatial} \\ 
\bottomrule
\end{tabular}
\caption{Our results for $\pwpar{1}$, $\pwpar{d}$ for $d\ge 2$, and $\wpwpar{1}$ and corresponding previous results for these problems.}
\label{tab:results}
\end{table*}

%% file: related_work.tex
In voting theory, partial information has been explored under various voting models. Konczak and Lang~\cite{konczak2005voting} introduced the \emph{partial order model}, where each voter’s preferences are specified as a partial order rather than a complete ranking. They also formulated the two fundamental problems of \emph{necessary winner} and \emph{possible winner}, which analyze the conditions under which candidates can be guaranteed or potentially elected given the incomplete preferences.
Betzler and Britta~\cite{betzler2010towards} established the computational complexity of $\pw$ within the partial order model for all scoring rules except for $(2, 1, \dots, 1, 0)$. Specifically, they show that $\pw$ is solvable in polynomial time under the plurality and veto voting rules, while for other scoring rules it is NP-complete. Baumeister and Rothe~\cite{baumeister2012taking} extended the hardness results to the $(2, 1, \dots, 1, 0)$ voting rule.

Baumeister et al.~\cite{baumeister2011computational} investigated two variants of the possible winner problem. The first, Possible co-Winner with respect to the Addition of New Candidates (PcWNA), investigates whether adding a limited number of new candidates can enable a designated candidate to win, proving NP-completeness for various scoring rules. The second, Possible Winner/co-Winner under Uncertain Voting Systems (PWUVS and PcWUVS), examines whether a candidate can win under at least one voting rule within a class of systems, with NP-completeness established under certain conditions.
Chakraborty and Kolaitis~\cite{chakraborty2021classifying} analyzed the possible winner problem in the partial chains model, where partial orders include a total order on a subset of their domains. They established that this restriction does not affect the complexity.
Kenig~\cite{kenig2019complexity} analyzed the problem under partitioned voter preferences, providing a polynomial-time algorithm for two-value scoring rules and proving NP-hardness for three or four distinct values.

Truncated voting rules (or {\em truncated ballots}) are used to 
simplify voting procedures.
Baumeister et al.~\cite{baumeister2012campaigns} study the complexity of determining a $\pw$ given truncated ballots. Yang~\cite{yang2017complexity} and Terzopoulou and Endriss~\cite{terzopoulou2021borda} studied elections under different variants of truncated Borda scoring rules.
Do{\u{g}}an and Giritligil~\cite{dougan2014implementing} investigated the likelihood of choosing the Borda outcome using a truncated scoring rule.

Weighted voting models, where voter influence is weighted, have been explored extensively \cite{bartholdi1989computational,brandt2016handbook,conitzer2002complexity,conitzer2007elections}. 
Pini et al.~\cite{pini2011incompleteness} studied $\nw$ and $\pw$ with weighted voters in the partial orders model, and showed NP-hardness results for Borda, Copland, Simpson, and STV rules. Walsh~\cite{walsh2007uncertainty} extended these results to cases where the number of candidates is bounded.
Baumeister et al.~\cite{baumeister2012possible} analyzed weighted $\pw$ where voter preferences are known but weights are unknown.

Spatial voting generalizes single-peaked preferences by embedding voters and candidates in a multidimensional space, where preferences are single-peaked along certain dimensions.
Single-peaked preferences, first studied by Black~\cite{black1948rationale}, have been widely analyzed for their simplifying effects on voting problems such as manipulation and winner determination under many voting rules~\cite{brandt2015bypassing,moulin1984generalized}.
Faliszewski et al.~\cite{faliszewski2009shield} show that NP-hardness of manipulation and control
vanishes under single-peak preferences. On the other hand, the hardness result remains for weighted elections. In Section~\ref{sec:weighted} we adjust some of these results for $\wpwpar{1}$.

%% file: prelims.tex
\subsection{Spatial Voting}
Let $V = \{v_1, \dots ,v_n\}$ denote the set of voters and $C = \{c_1, \dots ,c_m\}$ the set of candidates, where $m\geq 2$ to avoid trivial cases. Every candidate has a position, 
in the $d$-dimensional space 
representing their opinions on $d$ issues.\footnote{For the case of $d=1$, we assume $c_1<c_2<\dots <c_m$.} Each voter $v_i$ has a ranking $R_i$ over all candidates. The collection of all rankings for all the voters forms a \em{ranking profile}\em, denoted by $\mathbf{R}= (R_1, \dots , R_n)$.

A \em{spatial voting profile }\em $\mathbf{T}=(T_1,...,T_n)$ consists of $n$ points, where $T_j = \langle T_{j,1}, \dots, T_{j,d}\rangle \in \mathbb{R}^d$ represents voter $v_j$'s opinion on $d$ issues.
Given a spatial voting profile $\mathbf{T}$, $\mathbf{R}_{\mathbf{T}} = (R_{T_1},..., R_{T_n})$ is the derived ranking profile, where each voter $v_j$ ranks candidates in $C$ according to their distance from $v_j$'s opinion, $T_j$. The closest candidate is ranked first, and the farthest is ranked in position $m$ in $v_j$’s preferences. Tie breaking rule is arbitrary but fixed for all voters.

\input{figures/spatial_voting}

In Figure~\ref{fig:spatial_voting}, we illustrate spatial voting in a two-dimensional space. In this example, there are two voters, Alice and Bob, who are choosing a vacation destination. The candidates, representing the possible destinations, are Rio de Janeiro, New York, and Iceland. Each dimension corresponds to a decision criterion: $d_1$ represents urbanization, and $d_2$ represents temperature. Each candidate occupies a position in the space that reflects these criteria; for example, Iceland has a low temperature coordinate and a low urbanization coordinate, as it is a cold and rural destination. The positions of Alice and Bob in the space are denoted by $T_A$ and $T_B$, respectively. In this example, Alice's derived ranking profile is $R_{T_A} = (\text{Iceland}, \text{Rio de Janeiro}, \text{New York})$, based on the distances of the candidates from $T_A$.

\subsection{Voting Rules}
A \em{voting rule }\em is a function that maps a ranking profile to a nonempty set of winners. This paper focuses mainly on positional scoring rules, where candidates earn points based on their rank positions. 
A positional scoring rule $r$ is defined as a sequence $\{\vec{s}_m\}_{m \geq 2}$ of $m$-dimensional \emph{score vectors} $\vec{s}_{m} = (s_m(1), \dots, s_m(m))$. For each $m\in \mathbb{N}$ the vector 
$\vec{s}_{m}$ consists of $m$ natural numbers that satisfy \( s_m(1) \geq \dots \geq s_m(m) \) and \( s_1(m) > s_m(m) \).

For a ranking profile $\mathbf{R} = (R_1,...,R_n)$ and a positional scoring rule $r$ with a score vector $\vec{s}_m$, the score assigned to candidate $c$ by voter $v_j$ is $s(R_j,c)=s_m(i)$, where $c$ is ranked in the $i$-th position in $R_j$. The total score of candidate $c$ by ranking profile $\mathbf{R}$ is denoted by $s(\mathbf{R},c) = \sum_{j=1}^n s(R_j,c)$. 
Examples for positional scoring rules include plurality $(1,0,...,0)$, veto $(1,...,1,0)$, $k$-approval $(1, \dots, 1, 0, \dots, 0)$ where the number of '1' entries is $k$, and the Borda rule, defined with the scoring vector $(m-1,m-2,...,0)$.

A two-valued positional scoring rule consists of two values which are w.l.o.g 1 and 0.
Such rules can be described as \emph{$k(m)$-approval}, where for a number of candidates $m$, the $m$-dimensional score vector $\vec{s}_m$ consists of $k(m)$ '1' entries. Note that throughout the paper, when using the term $k$-approval, we refer to a $k$ which is not dependent on $m$, therefore fixed.

One focus of this paper is a subclass of positional scoring rules called \emph{truncated scoring rules} \cite{dougan2014implementing}. A \(k \)-truncated score vector has strictly positive values in exactly its first \( k \) entries. Thus, a \( k \)-truncated scoring rule allows voters to allocate score to exactly \( k \) candidates.

\subsection{Partial Spatial Voting}
Imber et al.~\cite{imber2024spatial} introduced the \em{partial spatial voting model}\em, where voters' preferences are incompletely specified. This model is represented by a \em{partial spatial profile }\em $\mathbf{P}=(P_1, \dots, P_n)$, where each voter $v_j$ is described as a vector of intervals $P_j = \langle [\ell_{j,1},u_{j,1}], \dots, [\ell_{j,d},u_{j,d}] \rangle$, and \( [\ell_{j,i}, u_{j,i}] \) represents the lower and upper bounds of $v_j$'s ideal point in each issue. The precise positions of the candidates are assumed to be known.

A spatial voting profile $\mathbf{T} = (T_1, \dots , T_n)$ is a \em{spatial completion }\em of $\mathbf{P}$ if, for every voter $v_j$, $T_{j,i} \in [\ell_{j,i},u_{j,i}]$. 
The ranking profile \( \mathbf{R_T}\) is then derived from this completion. A ranking profile $\mathbf{R}$ is a \em{ranking completion }\em of $\mathbf{P}$ if there exists a spatial completion $\mathbf{T}$ such that $\mathbf{R} = \mathbf{R_T}$.
\begin{definition}
    Given a partial profile $\mathbf{P}$ and a candidate $c^* \in C$, the \emph{possible winner problem} under a voting rule $r$ asks whether there exists a profile completion $\mathbf{T}$ of $\mathbf{P}$ such that $c^*$ is a winner w.r.t. $r$, i.e., $s(\mathbf{R}_{\mathbf{T}}, c^*) \geq s(\mathbf{R}_{\mathbf{T}}, c)$ for each $c \in C$.
\end{definition}

Figure~\ref{fig:partial_spatial} illustrates a partial spatial profile based on the example from Figure~\ref{fig:spatial_voting}.
Instead of precise positions in space, each voter is represented by lower and upper bounds on their opinion for each issue, forming a region of possible positions, depicted as orange rectangles.
$T_A$ and $T_B$ are spatial completions of the partial profile in which Iceland is ranked first by both voters.
Similarly, $T'_A$ and $T'_B$ are valid spatial completions where Rio de Janeiro is ranked first by both voters.

\input{figures/partial_spatial}

\subsection{Spatial Approval Voting}
In approval voting voters partition candidates into ``approved" and ``unapproved" groups, selecting the candidate with the highest approval count. 
Unlike $k$-approval, the number of approvals per voter varies. 
In spatial settings, each voter $v_j$ has an \emph{approval radius} $\rho_j \in \mathbb{R}$ and approves candidates within a distance $\rho_j$. 
Given a spatial completion $\mathbf{T}$, the approval set for voter $v_j$ is $A_{T_j} = \{c \in C : \|T_j - c\|_2 \leq \rho_j\}$. Approval regions correspond to intersections of $d$-dimensional spheres and the voter’s position rectangle.

\subsection{Parameterized Complexity}
We adopt the standard concepts and notations from parameterized complexity theory \cite{cygan2015parameterized,downey2013fundamentals,niedermeier2002invitation}. A parameterized problem $L \subseteq \Sigma^* \times \mathbb{N}$ is a subset of all instances $(x, k)$ from $\Sigma^* \times \mathbb{N}$, where $k$ represents the parameter. A parameterized problem $L$ is in the class FPT (fixed-parameter tractable) if there exists an algorithm that decides every instance $(x, k)$ of $L$ in $f(k) \cdot |x|^{O(1)}$ time, where $f$ is any computable function that depends solely on the parameter.

%% file: figures/spatial_voting.tex
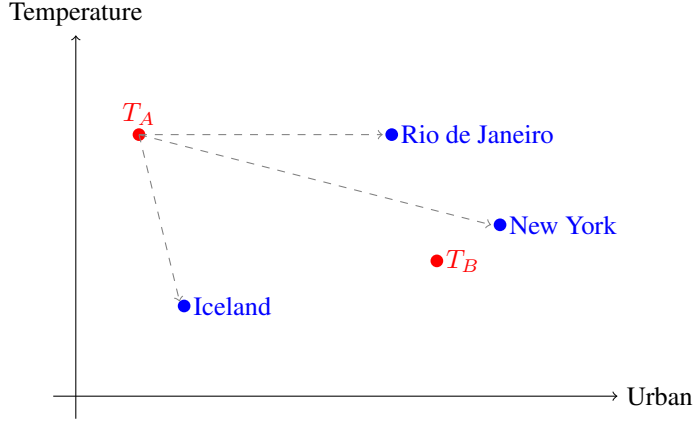
\begin{figure}[htbp!]
\centering
\begin{tikzpicture}[scale=1.2]
    % Axes
    \draw[->] (-0.25,0) -- (6,0) node[right] {Urban};
    \draw[->] (0,-0.25) -- (0,4) node[above] {Temperature};
    
    % Points
    \fill[blue] (1.2,1) circle (2pt) node[right] {Iceland};
    \fill[blue] (3.5,2.9) circle (2pt) node[right] {Rio de Janeiro};
    \fill[blue] (4.7,1.9) circle (2pt) node[right] {New York};

    \fill[red] (0.7,2.9) circle (2pt) node[above] {$T_A$};
    \fill[red] (4,1.5) circle (2pt) node[right] {$T_B$};

    \draw[->,gray,dashed] (0.7,2.9) -- (1.15,1.05);
    \draw[->,gray,dashed] (0.7,2.9) -- (3.4,2.9);
    \draw[->,gray,dashed] (0.7,2.9) -- (4.6,1.9);

\end{tikzpicture}
\caption{Example of spatial voting in a two-dimensional space.}
\label{fig:spatial_voting}
\end{figure}

%% file: figures/partial_spatial.tex
\begin{figure}[htbp!]
\centering
\begin{tikzpicture}[scale=1.2]
    % Axes
    \draw[->] (-0.25,0) -- (6,0) node[right] {Urban};
    \draw[->] (0,-0.25) -- (0,4) node[above] {Temperature};

    % Rectangles
    \fill[orange, opacity=0.35] (0.5,1.2) rectangle (3,3.2);
    \fill[orange, opacity=0.35] (2.5,0.5) rectangle (3.5,2.5);    

    % Dashed lines
    \draw[dashed,gray] (0.5,0) -- (0.5,4) node[left] {};
    \draw[dashed,gray] (3,0) -- (3,4) node[right] {};
    \draw[dashed,gray] (0,1.2) -- (6,1.2) node[above] {};
    \draw[dashed,gray] (0,3.2) -- (6,3.2) node[above] {};

    % Labels for dashed lines
    \node[left] at (0,3.2) {$u_{A,2}$};
    \node[left] at (0,1.2) {$\ell_{A,2}$};
    \node[below] at (0.5,0) {$\ell_{A,1}$};
    \node[below] at (3,0) {$u_{A,1}$};
    
    % candidates
    \fill[blue] (1.2,1) circle (2pt) node[right] {Iceland};
    \fill[blue] (3.5,2.9) circle (2pt) node[right] {Rio de Janeiro};
    \fill[blue] (4.7,1.9) circle (2pt) node[right] {New York};

    \fill[red] (0.9,1.3) circle (2pt) node[above] {$T_A$};
    \fill[red] (2.6,1.4) circle (2pt) node[above] {$T_B$};

    \fill[red] (2.9,3) circle (2pt) node[left] {$T'_A$};
    \fill[red] (3.3,2.2) circle (2pt) node[right] {$T'_B$};

\end{tikzpicture}
\caption{Illustration of a partial spatial profile and two different spatial completions.}
\label{fig:partial_spatial}
\end{figure}
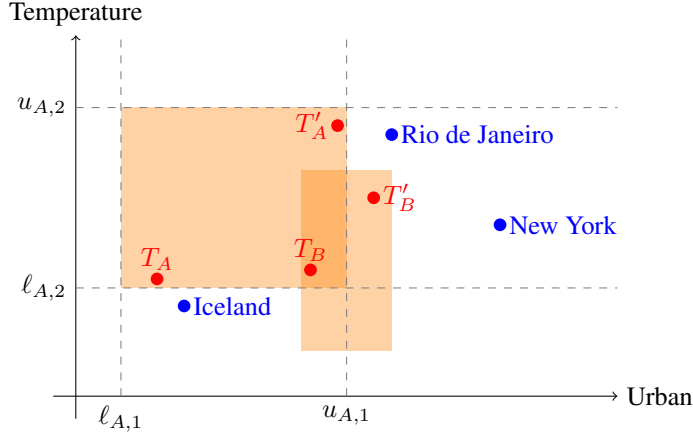

%% file: k-truncated.tex
This section establishes that $\pwpar{1}$ with any $k$-truncated voting rule can be solved in polynomial time when $k$ is constant. To do so, we introduce a new multi-machine scheduling problem, termed \emph{shapes scheduling}, where processing a job requires varying machine resources over time. We then provide a polynomial-time reduction from $\pwpar{1}$ to the shapes scheduling problem.
In the reduction, every voter becomes a job, and the resources used to process it reflect the score the voter hands to candidates.
Finally, we present a dynamic programming algorithm to efficiently solve shapes scheduling instances.

\subsection{The Shapes Scheduling Problem}

In \emph{shapes scheduling} each job may use multiple machines, in a quantity that changes over the processing time. Each scheduling option is referred to as a \emph{shape}, which specifies the number of machines required at any time throughout processing. Assume time is slotted. We first present the notion of a shape. 
Let $[r]$ denote the set $\{1, \dots, r\}$.
\begin{definition}
	Let $p\in\N$. A \emph{shape $\shape$ }is a vector $(M_{0}^{\shape},...,M_{p-1}^{\shape})$
	such that $M_{i}^{\shape}\in\N_{0}$ for each $i\in \{0\}\cup [p-1]$. We denote by
	$p$ the \emph{processing time }of $\shape$. 
\end{definition}
The intuition is that if job a $j$ is scheduled at time $t\in\N$
with a shape $\shape$ and a processing time~$p$ then for each $i\in \{0\}\cup [p-1]$,
during the interval $[t+i,t+i+1)$ job $j$ occupies $M_{i}^{\shape}$ machines. 
For instance, consider the shape $\shape = (2,1)$ with $p=2$. Figure~\ref{fig:shape} shows two ways to schedule the job at time 
$t$ using $\shape$, both satisfying the requirement of two machines during $[t,t+1)$ and one machine during $[t+1,t+2)$. Note that the machine indices are irrelevant, and there is no requirement to use the same machine across consecutive time slots. Additionally, preemption is not permitted.

\input{figures/shape_demonstration}

In the \emph{shapes scheduling problem} we are given a set of
$M$ identical machines for some $M\in\N$, and a set of jobs $J$. Each job $j\in J$ is associated with $(i)$ a processing time $p_{j}\in\N$, $(ii)$ a release time $r_{j}\in\N_0$, $(iii)$ a deadline $d_{j}\in\N$ with $r_{j}+p_{j}\le d_{j}$, and $(iv)$ a set of shapes $\H_{t}^{(j)}$, each with processing time $p_{j}$, for any time $t\in\N_0$ such that $r_{j}\le t\le d_{j}-p_{j}$.

The goal is to select for each job $j\in J$ a starting time $S_{j}\in\N_{0}$
satisfying $r_{j}\le S_{j}\le d_{j}-p_{j}$ and a shape $\shape^{(j)}\in\H_{t}^{(j)}$.
Given these starting times and shapes, for each time $t\in\N$
we denote the number of busy machines during $[t,t+1)$ by $M(t)$.
Formally, we define $M(t):=\sum_{j\in J:S_{j}\le t<S_{j}+p_{j}}M_{t-S_{j}}^{\shape^{(j)}}$.
We require for each $t\in\N_{0}$ that $M(t)\le M$, i.e., at most
$M$ machines are used during the interval $[t,t+1)$.

\subsection{Reduction from $\pwpar{1}$ to Shapes Scheduling}\label{sec:reduction}

We show how we can reduce $\pwpar{1}$ to the shapes scheduling problem. 
Given an instance of $\pwpar{1}$, the release times
and deadlines of our jobs will be in the interval $[1,m+1]$;
intuitively, for each $i\in[m]$
the interval $[i,i+1)$ corresponds
to candidate $c_{i}$. For each voter $v_{j}\in V$ we define a job
$j\in J$ as follows. We set $p_{j}=k$.
Let $i_{L}$ be the smallest
index such that candidate $c_{i_L}$ receives a positive score from $v_{j}$ if $T_{j}=\ell_{j}$. We set $r_{j}=i_{L}$.
Similarly, let $i_{R}$ be the largest index such that candidate
$c_{i_R}$ receives a positive score from $v_{j}$ if $T_{j}=u_{j}$.
We set $d_{j}=i_{R}+1$. We claim that $c_{i_L}$ is the leftmost candidate which $v_j$ can vote for and $c_{i_R}$ is the rightmost candidate which $v_j$ can vote for.
See Figure~\ref{fig:voting_range}.

\input{figures/voting_range}

\begin{lemma}\label{lem:vote-range}
For each possible position $T_{j}\in[\ell_{j},u_{j}]$ for voter $v_{j}$, only candidates in $\{c_{i_L},...,c_{i_R}\}$ receive a score from $v_{j}$.
\end{lemma}

\begin{proof}
Let $T_{j}\in[\ell_{j},u_{j}]$ be a position for voter $v_{j}$, and $c_i\in \{c_{1},...,c_{i_L-1}\}\cup\{c_{i_R+1},...,c_{m}\}$ a candidate. W.l.o.g., $i\in \{1, \dots, i_L-1\}$. We prove that $c_i$ does not receive votes from $v_j$.

By the nature of positional scoring rules and the definition of $i_L$, for $T_j=\ell_j$ the other candidates which receive a positive number of votes by $v_j$ are $c_{i_L+1}, \dots , c_{i_L+P-1}$. Then for every $c'\in\{c_{i_L}, \dots, c_{i_L+P-1}\}$, it holds that $|\ell_j-c'|<|\ell_j-c_i|$.
We note that $\ell_j\leq T_j$ and because $c_i<c_{i_L}$, $c_i<\ell_j$.
Then $|T_j-c_i| = T_j-c_i\geq \ell_j-c_i = |\ell_j-c_i|$.
Therefore, $|T_j-c_i|\geq |\ell_j-c_i|>|\ell_j-c'|$ for every $c'\in\{c_{i_L}, \dots, c_{i_L+P-1} \}$, promising that $c_i$ does not receive any votes from $v_j$.
\end{proof}

Next, we define the set of allowed shapes for $j$. Consider a value
$t\in[m]$ with $r_{j}\le t\le d_{j}-k$. Let $\mathcal{T}_{j,t}$ denote
the set of possible positions $T_{j}$ for $v_{j}$ such that exactly
the candidates $c_{t},...,c_{t+k-1}$ receive a score, meaning these candidates are the top $k$ candidates in $R_{T_j}$. For each $T_j\in \mathcal{T}_{j,t}$ and each $i\in\{0\}\cup[k-1]$, $s(R_{T_j},c_{t+i})$ is the score that candidate $c_{t+i}$ receives
from voter $v_{j}$ if it is positioned at $T_j$.
This yields a shape $(s(R_{T_j},c_{t}),...,s(R_{T_j},c_{t+k-1}))$. 

Figure~\ref{fig:shape_creation} illustrates two possible positions of voter $v_j$, denoted $T_j$ and $T'_j$.
Let the scoring rule be 2-truncated Borda: $\vec{s}=(2,1, 0 \dots, 0)$. At $T_j$, the ranking of $v_j$ is $R_{T_j}=(c_1, c_2, c_3)$, where the top two candidates are $c_1$ and $c_2$, placing $T_j\in \mathcal{T}_{j,1}$. As $s(R_{T_j},c_1)=2$ and $s(R_{T_j},c_2)=1$, the resulting shape is $\shape=(2,1)$.
At \( T'_j \), the ranking is \( R_{T'_j} = (c_3, c_2, c_1) \), making \( c_2 \) the lowest-indexed candidate in the top two. Therefore, \( T'_j \in \mathcal{T}_{j,2} \), resulting in the shape \( \shape' = (1, 2) \).

\input{figures/shape_creation}

We define $\H_{t}^{(j)}$
to be the set of all these shapes, i.e., $\H_{t}^{(j)}:=\left\{ (s(R_{T_j},c_{t}),...,s(R_{T_j},c_{t+k-1})): T_j\in \mathcal{T}_{j,t}\right\}$.
We can compute the set $\H_{t}^{(j)}$ by showing that there is a
subset of positions $T_j$ in $\mathcal{T}_{j,t}$ that suffice for defining all shapes in $\H_{t}^{(j)}$, and that
we can construct this subset efficiently.

\begin{lemma}
\label{lem:compute-shapes}For each voter $v_{j}$ and each $t\in[m]$
with $r_{j}\le t\le d_{j}-k$ we can compute the set $\H_{t}^{(j)}$
in time $O(knm^2)$. 
\end{lemma}

\begin{proof}
    For any pair of candidates $c_i<c_h$, the middle point $m_{i,h}=\frac{c_h-c_i}{2}$ separates the space into two regions: every voter $v_j$ whose position is $T_j \leq m_{i,h}$ prefers candidate $c_i$ over $c_h$, and every voter $v_j$ whose position is $T_j > m_{i,h}$ prefers $c_h$ over $c_i$.
    In this case, the tie breaking is in favor of the lower indexed candidate, though it can be adjusted to every fixed tie breaking rule.
    By finding the middle point for each pair of candidates, we separate the space into $\binom{m}{2}+1$ segments, where the ranking of candidates for all voters positioned are in a given segment are the same.

    For each segment $E$, denote by $R_E$ the ranking profile for voters positioned in segment $E$, i.e., $R_E=(c_{\ell_1},c_{\ell_2},\ldots, c_{\ell_m})$, where $c_{\ell_1}$ is the candidate who receives the highest number of votes, $c_{\ell_2}$ the second to highest, and so on. 
    Let $z_E$ be the smallest index of a candidate who is in the top $k$ candidates in $R_E$. Note that $z_E$ is a non-decreasing series by the segment going left to right.
    We define a shape $\shape(E) = (M_0^{\shape(E)}, \dots , M_{k-1}^{\shape(E)})$ for each segment as follows. For each $i\in \{0\}\cup[k-1]$, the $i$th entry in the shape vector, $M_i^{\shape(E)}$, is 
    the score candidate $c_{z_E+i}$ receives by the ranking profile of segment $E$, i.e. $M_i^{\shape(E)}=s(R_E,c_{z_E+i})$. 

    For each voter $v_{j}$ and each $t\in[m]$ with $r_{j}\le t\le d_{j}-k$ we can compute the set $\H_{t}^{(j)}$. First, we define for every $t\in [m]$, $\H_t = \{\shape(E)|~ \forall E: z_E=t\}$.
    Let $v_j \in V$ with $P_j = [\ell_j, u_j]$.
    \begin{itemize}
        \item For every $t$ such that $r_j<t<d_j-k$: $\H_{t}^{(j)}=\H_t$.
        \item $\H_{r_j}^{(j)} = \{\shape(E) |\; \forall E: z_E = r_j \wedge (E\cap [\ell_j, u_j]\neq \emptyset)\}$.
        \item $\H_{d_j}^{(j)} = \{\shape(E) |\; \forall E: z_E = d_j-k \wedge (E\cap [\ell_j, u_j]\neq \emptyset)\}$.
    \end{itemize}
    We analyse the complexity of computing the subsets. Computing all segments takes $O(m^2)$. Then, we compute all ranking profiles. The first ranking profile, generated by the leftmost segment, is $(c_1,c_2,\ldots, c_m)$. Moving to the next segment, note that when crossing the middle point between two candidates only the order between the two changes; therefore, the new ranking profile can be computed in $O(n)$. From the ranking profile, determining $z_E$ and $\shape(E)$ takes $O(k)$. Therefore, the first step takes $O(knm^2)$.
    For each job we calculate the subsets of shapes at $r_j$ and $d_j-k$, which can be done by iterating over all segments.
    All together, calculating all the subsets takes $O(knm^2+ nm^2) = O(knm^2)$.
\end{proof}

We illustrate the ideas behind the shape sets computation in Figure~\ref{fig:segment_split} with an example with four candidates, under the voting rule $\vec{s}_4 = (3,2,1,0)$ which is 3-truncated Borda ($k=3$).
Each dashed line labeled $m_{i,k}$ is the middle point between candidate $c_i$ and candidate $c_k$. The resulting segments are denoted by $E_1, \dots , E_7$. For each segment $E_i$, $z_{E_i}$ is the lowest index in the top $k$ candidates of the ranking profile $R_E$ for the segment. For example, consider the segment $E_2$ between $m_{1,2}$ and $m_{1,3}$. Then $R_{E_2} = (c_2, c_1, c_3, c_4)$, thus $z_{E_2}=1$. As for $\shape(E_2)$, $c_{z_{E_2+0}} = c_1$ is ranked second, therefore $M_0^{\shape(E_2)} = s(R_{E_2},c_{z_{E_2}+0}) =s(R_{E_2},c_{1+0}) = 2$.
Also, $c_{z_{E_2+1}} = c_2$ is ranked first, thus $M_1^{\shape(E_2)} = s(R_{E_2},c_{z_{E_2}+1}) =s(R_{E_2},c_{1+1}) = 3$.
Same for $c_3$, resulting in the shape $\shape(E_2) = (2,3,1)$.

\input{figures/segments_split}

We extend the example by adding a voter $v_1$ with $P_1 = [\ell_1, u_1]$, marked in red. We define the corresponding job. The lowest index of a candidate that can receive a score from $v_1$ is $c_1$, and the highest index of a candidate that can receive score from $v_1$ is $c_4$. Then by the reduction, $r_1 = i_L = 1$, as $\ell_1\in E_2$, and $d_1= i_R+1 = 4+1 = 5$ as $u_1\in E_6$. This restricts the job $J_1$ to use the machines between times $t=1$ and $t=4$, matching the set of candidates that can receive a score from $v_1$.

We continue with computing the subsets of shapes. $\H_t$ contains all shapes $\shape(E_i)$ such that $z_{E_i} = t$. For example, $\H_1 = \{ \shape(E_1), \shape(E_2), \shape(E_3), \shape(E_4) \}$. As for the endpoint subsets for the job $j_1$, which represents $v_1$, $\H_{r_j}^{(j_1)} = \{\shape(E_2), \shape(E_3), \shape(E_4)\}$, since among all segments for which $z_E = r_1 = 1$, these are the only segments that overlap with $P_1=[\ell_1,u_1]$. Similarly, $\H_{d_j}^{(j)} = \{\shape(E_5), \shape(E_6)\}$.

We proceed with the reduction. We constructed the set $\H_{t}^{(j)}$ for each job $j$ and each value $t\in[m]$ with $r_{j}\le t\le d_{j}-k$.
Now we show that for each job $j$, the possible starting
times and their associated shapes represent a possible
assignment of scores by voter $v_{j}$ to the candidates, depending on the position $T_{j}$ of $v_j$. 

\begin{lemma}
\label{lem:votes-schedule-correspondence-gen} For each job $j\in J$ there
is a starting time $S_{j}$ and a shape $\shape^{(j)}\in\H_{S_{j}}^{(j)}$
if and only if there is a position $T_{j}\in [\ell_j,u_j]$ for voter $v_{j}$
such that for every $i \in \{0\}\cup[k-1]$,
$v_{j}$ gives a score of $M_{i}^{\shape^{(j)}}$ to candidate
$c_{S_{j}+i}$. 
\end{lemma}

\begin{proof}
    We start with the first direction. Let $j\in J$ be a job scheduled at $S_j$ in shape $\shape^{(j)}\in\H_{S_{j}}^{(j)}$. We define an position $T_{j}\in [\ell_j,u_j]$ such that $v_{j}$ gives a score of $M_{i}^{\shape^{(j)}}$ to candidate $c_{S_{j}+i}$, i.e. $s(R_{T_j},c_{S_{j}+i}) = M_{i}^{\shape^{(j)}}$
    for all $i \in \{0\}\cup[k-1]$.
    As $\shape^{(j)}\in\H_{S_{j}}^{(j)}$, by the construction of $\H_{S_{j}}^{(j)}$ there exists a segment $E$ such that $E\cap P_j\neq \emptyset$, $z_E = S_j$ and $\shape(E) = \shape^{(j)}$. We define the position of $v_j$ to be a point $T_j\in E\cap P_j$, which is a valid because $T_j\in P_j$.
    Recall that the shape $\shape(E)$ is defined such that for each $i\in  \{0\}\cup[k-1]$, $M_i^{\shape(E)}$ is the number of votes given to $c_{z_E+i}$ by a voter positioned in $E$,
    as $M_i^{\shape(E)}=s(R_E,c_{z_E+i})$; therefore, for every such $i$, the number of votes given by $v_j$ to $c_{S_j+i}$ is the number of votes given to $c_{z_E+i}$,
    $s(R_E,c_{S_j+i}) = s(R_E,c_{z_E+i}) =M_i^{\shape(E)}$,
    i.e., $v_j$ gives $M_i^{\shape(E)}$ votes to candidate $c_{S_j+i}$.

    We continue with the second direction. Let $T_j\in [\ell_j, u_j]$ be a position for voter $v_j$ such that $v_{j}$ gives $s(R_{T_j},c_{S_{j}+i})$ votes to candidate $c_{S_{j}+i}$ for every $i\in \{0\}\cup[k-1]$.
    Let $E$ be the segment such that $T_j\in E$.
    We define the starting time of $j$, to be $S_j=z_E$ and the scheduling shape of $j$ to be $\shape(E)$, and prove $r_j\leq S_j\leq d_j-k$, $M_i^{\shape(E)}=s(R_{T_j},c_{S_{j}+i})$ for every $i\in \{0\}\cup[k-1]$ and $\shape(E)\in \H_{S_j}^{(j)}$.
    
    By Lemma~\ref{lem:vote-range}, if candidate $c_{S_{j}}$ receives a score from $v_j$, then $i_L\leq S_j$ where $c_{i_L}$ is the smallest candidate to receive a positive number, and $r_j=i_L\leq S_j$. Also, $S_j+k-1\leq i_R$ where $c_{I_R}$ is the largest candidate to receive a positive number, and $d_j = i_R+1\geq S_j+k$.
    By definition of $\shape(E)$, $M_i^{\shape(E)}=s(R_{T_j},c_{S_{j}+i-1})$.
    By the construction of the shape sets, for every segment $E \cap [\ell_j,u_j]$ with a scheduling time $r_j\leq S_j\leq d_j$, it holds that $\shape(E) \in \H_{S_j}^{(j)}$.
\end{proof}

The next step is to combine Lemma~\ref{lem:vote-range} and Lemma~\ref{lem:votes-schedule-correspondence-gen} to establish the correctness of the reduction.

\begin{lemma}
\label{lem:reduction-is-correct-gen}Let $i^*\in[m]$ be the index of candidate $c^*$. Then, candidate $c^{*}$ is a possible winner if and only if there is a number $M^{*}\in\{\sum_{i\in I} s_m(i) | \forall i\in I, i\in [k], |I| \leq n\}$
such that for the set of jobs $J$ there is a feasible schedule with
$M^{*}$ machines such that all machines
are busy during $[i^*,i^{*}+1)$.
\end{lemma}

\begin{proof}
    As $c^*$ is a possible winner, there is a spatial voting profile $\mathbf{T}=(T_1, \dots, T_n)$ and a value $M^*$ such that $s(\mathbf{R_T}, c^*)=M^*$, and for all $c \neq c^*$, $s(\mathbf{R_T}, c) \leq M^*$. $M^*$ is a sum of $n$ votes, therefore $M^*\in\{\sum_{i\in I} s_i | \forall i\in I, i\in [k], |I| \leq n\}$. Let $i^*\in [m]$ be the index of candidate $c^*$, i.e. $c_{i^*}=c^*$
We construct a feasible schedule $\mathcal{S}$ with $M^*$ machines such that $M(i^*)=M^*$.

By Lemma~\ref{lem:votes-schedule-correspondence-gen}, for every voter $v_j$ with position $T_j\in [\ell_j,u_j]$ there exists a starting time $S_j$ and a scheduling shape $\shape^{(j)}\in\H_{S_{j}}^{(j)}$ such that $v_{j}$ gives $M_{i}^{\shape^{(j)}}$  votes to candidate $c_{S_{j}+i}$ for $i \in \{0, \dots, P-1\}$. Let $\mathcal{S}$ be the schedule where each job $j$ is scheduled at $S_j$ in the corresponding shape.
Then, by Lemma~\ref{lem:votes-schedule-correspondence-gen}, $\mathcal{S}$ is feasible, i.e., for all $j$ we have $r_j \leq S_j \leq d_j - P$, and $j$ is scheduled in a shape corresponding to $S_j$. Moreover, the number of machines occupied by each job $j$ at each time slot $t$ is the number of votes candidate $c_t$ receives from $v_j$; 
therefore, the total number of busy machines at time slot $t$ is the total number votes candidate $c_t$ receives. As $s(\mathbf{R_T}, c) \leq M^*$ for any $c \neq c^*$, we have in $\cS$ that $M(t) \leq M^*$, and as $s(R_T, c^*) = M^*$, we have in $\cS$ that $M(i^*) = M^*$.

We now consider the other direction of the lemma. Let $\cS$ be a feasible schedule such that $M(i^*)=M^*$. We set for each voter $v_j$ a position $T_j \in [\ell_j,u_j]$ such that $c^*$ wins the election.
By Lemma~\ref{lem:votes-schedule-correspondence-gen}, for every job $j$ scheduled at $S_j$ in shape $\shape^{(j)}\in\H_{S_{j}}^{(j)}$, there is a possible position $T_{j}\in [\ell_j,u_j]$ for voter $v_{j}$ such that $v_{j}$ gives $M_{i}^{\shape^{(j)}}$ votes to candidate $c_{S_{j}+i}$, $i \in \{0, \dots, P-1\}$. $\mathbf{T}=(T_1, \dots, T_n)$ is a valid profile completion.
As before, the total number of busy machines at time slot $t$ is the total number of votes candidate $c_t$ receives. By the feasibility of the schedule, we have that $\forall c_t$ $M(t)\leq M^*$, therefore $\forall c \neq c^*$, $s(\mathbf{R_T}, c) \leq M^*$. As in $\mathcal{S}$ $M(i^*)=M^*$, it holds that $s(\mathbf{R_T}, c^*)=M^*$. This implies that $\forall c \neq c^*$, $s(\mathbf{R_T}, c^*) \geq s(\mathbf{R_T}, c)$,  making candidate $c^*$ a possible winner.
\end{proof}

The intuition behind this is that every use of machine at a time slot $[t,t+1)$ corresponds to a score given to candidate $t$ (Lemma~\ref{lem:votes-schedule-correspondence-gen}), therefore if all machines are busy at $c^{*}$, the schedule corresponds to a profile completion in which candidate $c^*$ receives $M^{*}$ votes, and no other candidate receives more, since there is only $M^{*}$ machines.

\subsection{An Algorithm for Shapes Scheduling}\label{sec:algorithm}

Our algorithm decides if there is a solution
to the shapes scheduling instance which satisfies Lemma~\ref{lem:reduction-is-correct-gen}. The algorithm exploits certain properties of the sets of possible shapes $\H_{t}^{(j)}$. To this end, we define the notion of \emph{$P$-structured} jobs.

\begin{definition}\label{def:p-structure}
Let $J$ be a set of jobs in an instance of shapes scheduling,
and let $P\in\N$. The set $J$ is $P$\emph{-structured} if 
\begin{itemize}
\item $p_{j}=P$ for each job $j\in J$,
\item for each $t\in\N_0$ there is a global set $\H_{t}$ such that if $t\neq r_j, d_j$, then $\H_{t}^{(j)}=\H_{t}$,
\item There exists an order of the jobs such that for every two jobs $j,j'\in J$ if $j\prec j'$ then either $d_j<d_{j'}$ or $d_j=d_{j'}$ and $\H_{d_j}^{(j)}\subseteq \H_{d_{j'}}^{(j')}$.
\end{itemize}
\end{definition}

Due to our construction of the instance $J$, we can show that for $P=k$, the jobs are $P$-structured.

\begin{lemma}\label{lem:exists_order}
    The job set $J$ generated by the reduction is $P$-structured for $P=k$. 
\end{lemma}

\begin{proof}
The first two properties in Definition~\ref{def:p-structure} hold trivially by the reduction. We prove that there exists an order of the jobs such that for any pair of jobs $j,j'\in J$, if $j\prec j'$ then either $d_j<d_{j'}$ or $d_j=d_{j'}$ and $\H_{d_j}^{(j)}\subseteq \H_{d_{j'}}^{(j')}$.

We order the jobs in non-decreasing order by the upper bounds of the corresponding voter intervals, i.e., in non-decreasing order of $u_j$, where ties are broken arbitrarily.
Let $j,j'$ be two jobs in a set of P-structured jobs $J$ with the corresponding voter intervals $P_j=[\ell_j, u_j]$ and $P_{j'}=[\ell_{j'}, u_{j'}]$. As $j\prec {j'}$, $u_j\leq u_{j'}$.
Assume towards contradiction the order does not satisfy the claim. Then there exists a shape $\shape\in \H_{d_j}^{(j)}$ such that $\shape \notin \H_{d_{j'}}^{(j')}$. Let $E=(e_1, e_2)$ be the segment such that $\shape(E)=\shape$. By the reduction, as $\shape \in \H_{d_j}^{(j)}$, $[\ell_j,u_j]$ overlaps with $(e_1,e_2)$, and as $\shape \notin \H_{d_{j'}}^{(j')}$, $[\ell_{j'},u_{j'}]$ does not overlap with $(e_1, e_2)$. This implies one of the following.
\begin{enumerate}
	\item[(i)]
	If $u_{j'}<e_1$, then because $u_j\leq u_{j'}$, we have $u_j<e_1$, in contradiction to the overlap of $[\ell_j,u_j]$ and $E=(e_1,e_2)$.
	\item[(ii)]
	If $e_2<\ell_{j'}$ then $(e_1,e_2)$ does not overlap with $[\ell_{j'}, u_{j'}]$, in contradiction to $\shape =\shape(E)\in \H_{d_{j'}}^{(j')}$.
\end{enumerate}
Both cases lead to a contradiction; therefore, the order satisfies the claim.
\end{proof}

We now present an algorithm for any $P$-structured instance of scheduling with shapes. Given a set of $P$-structured jobs $J$, a candidate $c^{*}\in C$, and a number of machines $M^{*}$, our algorithm decides if there exists a schedule for $J$ with $M^{*}$ machines such that all machines are busy at time $c^*$. Then, we run the algorithm for every possible value of $M^{*}$. This can be done in polynomial time since $M^*$ must be a combination of $n$ votes, each of value $s_m(1),s_m(2), \dots, s_m(k)$ or 0.
The heart of our algorithm is formalized as Lemma~\ref{lem:generalized-Baptiste},
which generalizes a result of \cite{baptiste2000scheduling}. Intuitively, our lemma states
that if there is a feasible schedule with $M^{*}$ machines, then
there is also a feasible schedule in which a job $j'$ with the latest deadline among all jobs in $J$
starts at a time $S_{j'}$ such that the remaining jobs are split nicely into two parts:
\begin{itemize}
\item a set $J_{L}$ containing all jobs $j\in J\setminus\{j'\}$ with
$r_{j}<S_{j'}$ and for each job $j\in J_{L}$ we have $S_{j}\le S_{j'}$,
and
\item a set $J_{R}$ containing all jobs $j\in J\setminus\{j'\}$ with
$r_{j}\ge S_{j'}$; thus, for each job $j\in J_{R}$ we
have $S_{j}\ge S_{j'}$.
\end{itemize}
This allows to partition our problem into two independent
subproblems, one for $J_{L}$ and one for $J_{R}$, on which we recurse.
We define a total order $\prec$ for the jobs in $J$
such that for any two jobs $j,j'\in J$ we have $j\prec j'$
if $d_{j}<d_{j'}$, or if $d_{j}=d_{j'}$ and $\H_{d_{j}}^{(j)}\subseteq \H_{d_{j'}}^{(j')}$.
Such order exists by Lemma~\ref{lem:exists_order}.
Using this order, we define the notation $U_{j'}(t,t')$ for subsets of jobs that we use below.

\begin{definition}
For any $j'\in J$ and $t,t'\in \N_0$, let
$U_{j'}(t,t')=\left\{j\;  |\; (j\preceq j') \wedge (t\leq r_j < t')\right\}$.
\end{definition}

Note that if $j'$ is the last job in the total order $\prec$ among all jobs in $J$ then $J=U_{j'}(0,d_{j'})$.
We formalize the partition of our problem into two independent subproblems.

\begin{lemma}
\label{lem:generalized-Baptiste}
Consider an instance of shapes
scheduling with a set of $P$-structured jobs $U_{j'}(t,t')$ where $j'\in U_{j'}(t,t')$,
and let $j''\in U_{j'}(t,t')$ such that $j \prec j''$.
Assume there is a schedule with a corresponding value $M(t)$ for each $t\in\N_0$.
Then there exists also a schedule with job start times
$\left(S_{j}\right)_{j\in J}$, the same value $M(t)$
for each $t\in\N$, and a partition of $U_{j'}(t,t')$ into three sets $\{j'\},J_{L},$and $J_{R}$ such that 
\begin{itemize}
\item $J_{L}=\{j\in U_{j''}(t,t'):r_{j}<S_{j'} \} = U_{j''}(t,S_{j'})$ and
$S_{j}\le S_{j^{'}}$
for each job $j\in J_{L}$, and
\item $J_{R}=\{j\in U_{j''}(t,t'):r_{j}\ge S_{j'} \}= U_{j''}(S_{j'},t')$
and $S_{j}\ge S_{j'}$ for each job $j\in J_{R}$.
\end{itemize}
\end{lemma}

\begin{proof}
Let $\cal{S}$ be the schedule of the job set $J$ with the same value $M(t)$ for each $t\in\N$, in which the starting time of job $j'$, $S_{j'}$, is maximal. We prove that in $\mathcal{S}$, $\forall j\in J_L:\; S_j\leq S_{j'}$ and $\forall j\in J_R:\; S_j\geq S_{j'}$.
	
Assume towards contradiction that the schedule does not satisfy the claim; then, at least one of the following occurs:
If there exists a job $j\in J_R$
such that $S_j<S_{j'}$ then $S_j<S_{j'}\leq r_j$, in contradiction to $\cal{S}$ being a feasible schedule. On the other hand, if 
there exists a job $j\in J_L$ such that $S_j>S_{j'}$, we construct a new schedule $\cal{S}'$ in which every job except for $j,j'$ is scheduled as in $\mathcal{S}$, and the remaining two jobs $j$ and $j'$ are swapped. Denote the starting times of the jobs in $\mathcal{S}'$ by $\left(S'_{j}\right)_{j\in J}$.
In the schedule $\cal{S}'$, for all $t\in \N$ $M(t)$ is the same as in $\mathcal{S}$.    
As for the remaining conditions for feasibility: for $j$ we have that $S'_j = S_{j'}< S_j \leq d_j-P$, and $S'_j = S_{j'} > r_j$ since $j\in J_L$; thus, the new start time $S'_j$ is valid. As $j$ is not scheduled at its release time or deadline, $\H_{S'_j}^{(j)}= \H_{t}$, i.e., the new scheduling shape is in the subset. 
For $j'$, recall that the jobs are sorted in non-decreasing order by their deadline, where in case of equal deadlines, i.e., $d_{j}=d_{j'}$, we have $\H_{d_{j}}^{(j)}\subseteq \H_{d_{j'}}^{(j')}$.
Job $j'$ is the last in the order out of all remaining jobs, therefore $j \prec j'$.
This implies that $S'_{j'}= S_j\leq d_j-P\leq d_{j'}-P$. On the other hand, $S'_{j'}=S_j>S_{j'}\geq r_{j'}$. 
As before, the new start time of $j'$ is valid. As $j'$ is not scheduled at its release time, the shape condition is satisfied for any start time except its deadline. 
If $j'$ is scheduled at its deadline then $S'_{j'} = S_j = d_j-P = d_{j'}-P$, implying that $d_j = d_{j'}$. As $j$ is scheduled at the end of its interval, it has a scheduling shape in the subset $ \H_{d_{j}}^{(j)}$. 
Since $\H_{d_{j}}^{(j)}\subseteq \H_{d_{j'}}^{(j')}$, it is in the subset of shapes allowed for $j'$ at its deadline. Hence, $\cal{S'}$ is a feasible schedule.
This contradicts the premise that $\cal{S}$ is the schedule with the same corresponding value $M(t)$ for each $t\in\N$.
\end{proof}

Assume that in our given instance, job $j'$ is last in the total order $\prec$ of $J$.
Algorithmically, we guess $S_{j'}$ in polynomial
time (as there are only a polynomial number of options). Once we guess $S_{j'}$ correctly, we directly obtain
$J_{L}$ and $J_{R}$. Note that during each time
interval $[t,t+1)$ with $t\in\N_0$ and $t<S_{j'}$ we can process
only jobs from $J_{L}$. On the other hand, during each time interval
$[t',t'+1)$ with $t'\in\N_0$ and $t'\ge S_{j'}+P$ we can process
only jobs from $J_{R}$. During $[S_{j'},S_{j'}+P)$
we may process jobs from $J_{L}$ but possibly also jobs from $J_{R}$.
Therefore, we also guess how to split the available machines between these two job sets during these time intervals.
Formally, we define $M_L(t)$ to be the number of machines allocated to $J_L$ at time $t$, and similarly $M_R(t)$ to be the number of machines allocated to $J_R$ at time $t$.
It must hold that $M_L(t)+M_R(t)+M_t^\shape \leq M^*$ for any $t$;
therefore, we guess $M_{L}(S_{j'}), \dots, M_{L}(S_{j'}+P-1)$, and assign the remaining machines to $J_R$ during $[S_{j'}, S_{j'}+P)$, i.e., we define
$M_{R}(S_{j'}+i):=M^*-M_{L}(S_{j'}+i)-M_i^\shape$ for any $i\in \{0\}\cup[P-1]$.
Each value of $M_L(t)$ is a combination of $n$ votes, therefore belongs to the set $\{\sum_{i\in I} s_m(i) | \forall i\in I, i\in [k], |I| \leq n\}$, meaning it has  $\binom{n+k}{n}=O(n^k)$ options.

This yields independent subproblems for $J_{L}$ and $J_{R}$ on which we recurse. To ensure that our running time is bounded by a polynomial in the input size,
we embed this recursion into a dynamic program with a polynomial number of DP-cells.
Each subproblem is associated with an interval
$[t,t')$ and a job $j'$, and we want to schedule the jobs $j\prec j'$ that are
released during $[t,t')$, i.e. $U_{j'}(t,t')$. During $[t,t+P)\cup[t',t'+P)$ we may not
have all $M^{*}$ machines available, as during these intervals
our subproblem may interact with other (previously defined) subproblems.
The DP-cell specifies how many machines are available during these intervals. 

Formally, each DP-cell is defined by a tuple $(j',t,t',M_{t},\dots,M_{t+P-1},M_{t'},\dots,M_{t'+P-1})$
such that 
\begin{itemize}
\item the values $t,t'\in\N_0$ define an interval $[t,t')$,
\item $j'\in J$ is the last job according to $\prec$ of the input jobs of the subproblem,
\item the values $M_{t},\dots, M_{t+P-1},\in\{M^*-\sum_{i\in I} s_m(i) | \forall i\in I, i\in [k], |I| \leq n\}$ 
denote the number of available machines during $[t,t+1), \dots, [t+P-1,t+P)$.
\item the values $M_{t'},\dots, M_{t'+P-1}\in\{\sum_{i\in I} s_m(i) | \forall i\in I, i\in [k], |I| \leq n\}$ 
denote the number of available machines during
$[t',t'+1),\dots, [t'+P-1,t'+P)$; note that the time points $t,\dots, t+P-1,t',\dots,t'+P-1$ may not be pairwise distinct.
\end{itemize}

Recall $M(t)$ denotes the number of busy machines during $[t,t+1)$.
The goal of this subproblem is to compute a schedule for the jobs
$U_{j'}(t,t')$ such that $M(t+i)\le M_{t+i}$ for any $i\in \{0\}\cup[P-1]$, and $M(t'')\le M^{*}$ for each $t''\in\N_0$
with $t+1<t''<t'$. 
For $i^*$ being the index of candidate $c^*$, if $i^{*}\in\{t+P,...,t'-P+1\}$ we require that
$M(i^{*})=M^{*}$; otherwise, we require that $M(i^{*})=M_{i^{*}}$.
Observe that the cell $(n,1,m+1,M^*,M^*,M^*,M^*)$
corresponds to the main problem we want to solve, where $n$ is the last job in the order $\prec$ of $J$.
Based on these DP-cells, we can construct a dynamic program which decides whether there exists a feasible schedule for the given set of jobs.

We demonstrate a scheduling of a job and the partitioning of the remaining resources into two subproblems. Figure~\ref{fig:schedule_example}, shows a scheduling subproblem with 4 machines. Each cell represents a machine in a specific time slot. The orange and green cells represent machines which are in use in another subproblem, and the rest of the cells are available for this subproblem. The job $j'$ is scheduled at time $S_{j'}=4$ in shape $\shape = (2,1)$, using two machines between $[4,5)$ and one machine between $[5,6)$. The tuple defining the subproblem is $(j', 1, 7, 2, 1,2,0)$. Indeed, as shown in the figure, at time $[t,t+1)=[1,2)$ there are 2 available machines, making $M_1 = 2$. Similarly, $M_2 = 1$, $M_7=2$, and $M_8=0$. The subset addressed is $U_{j'}(1,7)$.

\input{figures/schedule_example}

After scheduling job $j'$ at time 4, it becomes necessary to determine a fitting partition of the remaining available machines.
In our framework, the time slot needed to partitioned between $[S_{j'}, S_{j'}+2) = [4,6)$, which is marked in the figure between the dashed pink line. Indeed, by Lemma~\ref{lem:generalized-Baptiste}, the jobs in $J_L$ can be scheduled no later than $S_k = 4$, thus utilizing resources up to time $S_{j'}+2$. Concurrently, jobs in $J_R$ may also commence at $S_{j'}$.

We illustrate in one such partitioning scenario when we allocate to $J_R$ one machine from $[4,5)$ and one machine from $[5,6)$, therefore $M_L(S_{j'}) = M_L(S_{j'}+1) = 1$. The remaining machines for scheduling $J_L$ by this partition is illustrated in sub-figure (a) in Figure \ref{fig:after_partition}.
Now lets look on the other side, what remains to $J_R$. At each time, the remaining machines are all machine without the ones used to $j'$ and $J_L$, therefore $M_R(S_{j'}) = M^*-M_L(S_{j'})-M_0^{L} = 4-1-2=1$. At time $S_{j'}+1=5$, $M_R(S_{j'}+1) = M^*-M_L(S_{j'}+1)-M_1^{L} = 4-1-1=2$. These values, of the remaining machines for scheduling $J_L$ are illustrated in sub-figure (a) in Figure \ref{fig:after_partition}.

\begin{figure}[ht]
    \centering
    \subfloat[\centering Resource profiles of left subproblem.]{\input{figures/schedule_example_left}}
    \qquad
    \subfloat[\centering Resource profiles of right subproblem.]{\input{figures/schedule_example_right}}
    \caption{Resource profiles after the partition due to Figure~\ref{fig:schedule_example}.}
    \label{fig:after_partition}
\end{figure}

\begin{lemma}
\label{lem:algorithm-gen}Assume we are given an instance of the shape
scheduling problem with a set of $P$-structured jobs, $M^*$ machines and a candidate $c^{*}\in C$ with index $i^*\in[m]$.
There is an algorithm with a running time of $O(n^{1+3P^2}\cdot m^3)$ which decides
whether the instance admits a feasible schedule in which all machines are busy during $[i^{*},i^{*}+1)$.
\end{lemma}

\begin{proof}
Given $M^*$ machines, we run the DP to determine if there exists a feasible schedule using $M^*$ machines at time $c^*$.
If there is such a schedule then, by Lemma~\ref{lem:reduction-is-correct-gen}, $c^*$ is a possible winner.
We formulate a recursion to fill each DP-cell. As before, we renumber the jobs by their position in the total order $\prec$.
To shorten the notation, we define $\bar{M}_t^P:= M_t, \dots, M_{t+P-1}$.

For the DP-cell $(j',t,t',\bar{M}_t^P,\bar{M}_{t'}^P)$:
\begin{itemize}
    \item If $j=0$ and $c^{*}\notin\{t,...,t'+P-1\}$, or if $c^{*}\in\{t,\dots, t+P-1,t',\dots, t'+P-1\}$ and $M(c^{*})=0$:\:  $$B(j',t,t',\bar{M}_t^P,\bar{M}_{t'}^P)=0$$
    \item If $j=0$ and $c^{*}\in\{t+P,...,t'-1\}$, or if  $c^{*}\in\{t,\dots, t+P-1,t',\dots, t'+P-1\}$ and $M(c^{*})>0$:\:
    $$B(j',t,t',\bar{M}_t^P,\bar{M}_{t'}^P)= -\infty$$
    \item If $r_{j'}\notin [t,t')$:\: $$B(j',t,t',\bar{M}_t^P,\bar{M}_{t'}^P)=B(j'-1,t,t',\bar{M}_t^P,\bar{M}_{t'}^P)$$
    \item If there exists $i,i'\in \{0, \dots, P-1\}$ such that $t+i = t' +i' $ and $M_{t+i} +M_{t'+i'}>M^*$,
    then $$B(j',t,t',\bar{M}_t^P,\bar{M}_{t'}^P)=-\infty$$
    \item Otherwise:
    \begin{equation}\label{ex:DP-gen}
    \begin{aligned}
        B(j', t, t', \bar{M}_t^P,\bar{M}_{t'}^P) = \\
        \max_{\substack{
            S_{j'}:\; r_{j'} \leq S_{j'} \leq t'-P, \\
            \shape^{(j')} \in \H_{S_{j'}}^{(j')}, \\
            \forall i\in \{0, \dots, P-1\}:\; M_L({S_{j'}+i}) \leq M^*- M_i^{\shape^{(j')}}}}
        \Bigg( & B(j'-1, t, S_{j'}, \bar{M}_t^P, M_L({S_{j'}}),\dots, M_L({S_{j'}+P-1})) \\
        & + B(j'-1, S_{j'}, t', M_R({S_{j'}}), \dots, M_R({S_{j'}+P-1}), \bar{M}_{t'}^P)\\
        & + b(S_{j'}, \shape^{(j')}) \Bigg)
    \end{aligned}
    \end{equation}

    The maximum over an empty set is $-\infty$.
\end{itemize}

We prove the correctness of the dynamic program for each case.

If $j=0$ then the job set is empty. In case $c^{*}\notin\{t,...,t'+P-1\}$, or
$c^{*}\in\{t,\dots,t+P-1,t',\dots, t'+P-1\}$ and $M(c^{*})=0$, we return an empty schedule; therefore, the value of the objective function is zero, and $(j',t,t',\bar{M}_t^P,\bar{M}_{t'}^P)=0$. In the second case, any schedule using this subschedule would not reach the desired value of the objective function, therefore $(j',t,t',\bar{M}_t^P,\bar{M}_{t'}^P)=-\infty$.

For the case where $r_{j'}\notin [t,t')$, the subset of jobs in the subproblem is $U_{j'}(t,t')=U_{j'-1}(t,t')$; therefore, by definition $B(j',t,t',\bar{M}_t^P,\bar{M}_{t'}^P)=B(j'-1,t,t',\bar{M}_t^P,\bar{M}_{t'}^P)$.

We address the fourth case. If there exists $i,i'\in \{0, \dots, P-1\}$ such that $t+i = t' +i' $ and $M_{t+i} +M_{t'+i'}>M^*$ then the number of available machines at time $t+i$ is higher than the total number of machines, making it an infeasible schedule; therefore, $B(j',t,t',\bar{M}_t^P,\bar{M}_{t'}^P)= - \infty$.

We now prove the equality in the last case. Denoting the expression in the RHS of (\ref{ex:DP-gen} to be $B'$), we first prove that $B'\geq B(j'-1,t,t',\bar{M}_t^P,\bar{M}_{t'}^P)$. In the scenario where no feasible schedule exists, $B(j'-1,t,t',\bar{M}_t^P,\bar{M}_{t'}^P) = -\infty$, thus the inequality holds. We address the case where a feasible schedule exists, therefore $B(j'-1,t,t',\bar{M}_t^P,\bar{M}_{t'}^P)$ is finite. Let $\mathcal{S}$ be an optimal schedule for $U_{j'}(t,t')$ using the remaining available machines as inferred from $\bar{M}_t^P,\bar{M}_{t'}^P$. Let $S_{j'}$, $r_{j'}\leq S_{j'} \leq t'$ be the start time of $j'$ in $\mathcal{S}$, where $j'$ is the largest indexed job in the subproblem.
By Lemma~\ref{lem:generalized-Baptiste}, for any $j \in U_{j'-1}(t, S_{j'})$, it holds that $S_j \leq S_{j'}$, and for every $j \in U_{j'-1}(S_{j'}, t')$, $S_j \geq S_{j'}$. In other words, only jobs from $U_{j'-1}(t, S_{j'})$ can use machines before $t$ and only jobs from $U_{j'-1}( S_{j'},t')$ can use machines starting at $S_{j'}+P$.

Consider the time interval $[S_{j'}, S_{j'}+P)$ in which machines are allocated for job $j'$ and both jobs from $U_{j'-1}(t, S_{j'})$ and $U_{j'-1}( S_{j'},t')$ can use machines. In terms of $j'$, for $i\in \{0, \dots, P-1\}$, $M_i^\shape$ machines are allocated at $[S_{j'}+i,S_{j'}+i+1)$.
Let $M_L(S_{j'}+i)$ be the number of machines that jobs from $J_L$ are using at $[S_{j'}+i, S_{j'}+i+1)$. Each must be not greater than $M^*$ minus the number of machines $j'$ is using at that time, therefore $M_L(S_{j'}+i)\leq M^* - M_i^{\shape^{(j')}}$. This means that scheduling $U_{j'-1}(t, S_{j'})$ as in $\mathcal{S}$ yields a feasible schedule for $$(j'-1, t, S_{j'}, M_{t},\dots, M_{t+P-1}, M_L({S_{j'}}),\dots, M_L({S_{j'}+P-1})).$$

Next, consider the resources left for $J_R$. For each $i\in \{0,\dots, P-1\}$, reducing from the total number of machines the ones used for $j'$ and for $J_L$ leaves $M^* - M_L(S_{j'}+i)-M_i^{\shape^{(j')}}$ machine for $J_R$, making the schedule of $U_{j'-1}(t, S_{j'})$ as in $\mathcal{S}$ feasible for 
$$(j'-1, S_{j'}, t', M^* - M_L({S_{j'}}) - M_0^{\shape^{(j')}},\dots, M^* - M_L({S_{j'}+P-1}) - M_{P-1}^{\shape^{(j')}}, M_{t'},\dots, M_{t'+P-1}).$$
Then,

\begin{equation*}
\begin{aligned}
B(j', t, t', \bar{M}_{t'}^P,\bar{M}_{t'}^P) & = M_\mathcal{S}(c^*) \\
& = M_L(c^*) + M_R(c^*) + b(S_{j'}, \shape^{(j')}) \\
& \leq B(j'-1, t, S_{j'}, \bar{M}_{t}^{P}, M_L(S_{j'}), \dots, M_L(S_{j'}+P-1)) \\
& \quad + B(j'-1, S_{j'}, t', M_R(S_{j'}), \dots,M_R(S_{j'}+P-1), \bar{M}_{t'}^P) \\
& \quad + b(S_{j'}, \shape^{(j')}) = B'
\end{aligned}
\end{equation*}

Now, we prove that $B'\leq B(j',t,t',\bar{M}_t^P,\bar{M}_{t'}^P)$. Suppose that $B'$ is finite, otherwise the proposition holds trivially. Let $S_{j'}$ be the largest value such that for some $\shape^{(j')} \in \H_{S_{j'}}^{(j')}$, and $M_L({S_{j'}}),\dots, M_L({S_{j'}+P-1})$, brings the expression to a maximum.
There exists a schedule $\mathcal{S}_L$ that realizes 
$$B(j'-1, t, S_{j'}, \bar{M}_{t}^P, M_L({S_{j'}}), \dots, M_L({S_{j'}+P-1})),$$ 
and a schedule $\mathcal{S}_R$ that realizes 
$$B(j'-1, S_{j'}, t', M_R({S_{j'}}),\dots, M_R({S_{j'}+P-1}), \bar{M}_{t'}^P).$$ Note that every job in $U_{j'-1}(t,t')$ is scheduled either in $\mathcal{S}_L$ or in $\mathcal{S}_R$.

Consider the schedule $\mathcal{S}$ constructed as follows: Schedule $j'$ at time $S_{j'}$ in shape $\shape^{(j')}$, and schedule all other jobs in $U_{j'}(t,t')$ as in $\mathcal{S}_L$ or $\mathcal{S}_R$. We prove that $\cal{S}$ is a feasible schedule of $U_{j'}(t,t')$. Given $r_{j'} \leq S_{j'}$ and the feasibility of $\mathcal{S}_L$ and $\mathcal{S}_R$, all jobs adhere to their release and due date constraints.

We analyze different time points to ensure that at any time, no more than $M^*$ machines are used:

\begin{enumerate}
    \item For $t'' < S_{j'}$: Only jobs from $\mathcal{S}_L$ are scheduled, and since this is a feasible schedule, no more than $M^*$ machines are used, and specifically, at $t''\in\{t,\dots,t+P-1\}$ no more than $M_{t''}$ machines are used.
    \item For $t'' \geq S_{j'} + P$: Only jobs from $\mathcal{S}_R$ are scheduled. Similar to the previous case, because $\mathcal{S}_R$ is feasible no more than $M^*$ machines are used, and at $t''\in\{t',\dots, t'+P-1\}$, no more that $M_{t''}$ machines are used.
    \item At time slot $t''=S_{j'}+i$ for $i\in\{0, \dots, P-1\}$:
    $\mathcal{S}_L$ uses no more than $M_L(S_{j'}+i)$ machines,
    $\mathcal{S}_R$ uses no more than $M_R({S_{j'}+i}) = M^*-M_L(S_{j'}+i)-M_i^{\shape^{(j')}}$ machines,
    and $j'$ uses $M_i^{\shape^{(j')}}$ machines exactly.
    Overall, we have that
    \[
    \begin{array}{ll}
    M(S_{j'}+i) & \leq M_L(S_{j'}+i)+ M_R(S_{j'}+i) + M_i^{\shape^{(j')}} \\
    & \leq M_L(S_{j'}+i)+ M^*-M_L(S_{j'}+i)-M_i^{\shape^{(j')}}+M_i^{\shape^{(j')}} \\
    &  = M^* .
    \end{array}
    \]
\end{enumerate}

By the above, $\mathcal{S}$ is a feasible solution; therefore, $B'\leq B(j'-1,t,t',\bar{M}_t^P,\bar{M}_{t'}^P)$, which completes the proof of correctness of the DP.

We analyse the time complexity of the algorithm, starting with the number of DP-cells. There are $n$ jobs and $m$ different time options for $t$ and $t'$. The $2P$ other values in the tuple represent the number of available machines, which is bounded by $M^*= O(n^P)$; therefore, the total number of DP cells is $O(n^{1+2P^2} m^2)$.
The number of shapes depends on $P$ only, and  thus remains a constant.
The time complexity for calculating each DP-cell is $O(m\cdot n^{P^2})$, where $m$ is the factor of times the schedule $j'$ and $n^{P^2}$ is the resource allocation, a factor of $n$ for each time slot separating between the job sets $J_R$ and $J_L$.
Overall, running the DP for a possible value $M^*$ takes $O( n^{1+2P^2}\cdot m^2 \cdot m\cdot n^{P^2})= O(n^{1+3P^2}\cdot m^3)$.
\end{proof}

Now, Lemmas~\ref{lem:reduction-is-correct-gen} and \ref{lem:algorithm-gen}
imply the next result.

\begin{theorem}\label{thm:k-truncated}
We can solve the possible winner problem for any $k$-truncated voting rule
in time $O(n^{1+k+3k^2}\cdot m^3)$.
\end{theorem}

\begin{proof}
We start by calculating the shape scheduling instance induced by the voting instance. Then, we want to determine for candidate $c^*$, whether it is a possible winner. To this end, for each possible value of $M^*$, we conclude whether there exists a feasible schedule using $M^*$ machines at time $c^*$.
We use the DP to solve the corresponding scheduling problem.
If $(n,1,m+1,M^*, \dots ,M^*)=M^*$ for some $M^*$, then by Lemma~\ref{lem:algorithm-gen} there exists a schedule for the job instance such that at time $c^*$, all the machines are busy, and at any other time, not more than $M^*$ machines are busy. By Lemma~\ref{lem:reduction-is-correct-gen}, candidate $c^*$ is a possible winner.

If for every value of $M^*$, $(n,1,m+1,M^*,\dots,M^*)\neq M^*$, then by Lemma~\ref{lem:algorithm-gen} there is no schedule for the job instance such that at time $c^*$, all the machines are busy. By Lemma~\ref{lem:reduction-is-correct-gen}, candidate $c^*$ can not be a possible winner.
    
Each possible value for $M^*$ is a combination of $n$ values from the scoring vector, for which there are $\binom{n+k}{k}$ combinations. Therefore, for a constant $k$, determining for a candidate whether or not it is a possible winner takes $O(n^{1+k+3k^2}\cdot m^3)$.
\end{proof}

\subsection{Hardness Results for Scheduling with Shapes}\label{sec:hardness_for_scheduling}

In our algorithm from the previous subsection we required the input jobs to be $P$-structured which allowed us to solve the problem exactly in polynomial time for constant $P$. In this subsection, we complement this by showing that scheduling with shapes is strongly NP-hard if we lift these requirements. First, we remove the assumption that for each $t\in \N$ there is a global set $\H_t$ for the job shapes and require only that $P=O(1)$. In fact, we prove even that already for $P=1$ the problem is strongly NP-hard.

\begin{theorem}\label{thm:hardness-1}
	The scheduling with shapes problem is strongly NP-hard, even if $p_{j}=1$
	for each job $j\in J$.
\end{theorem}

\begin{proof}
We reduce from the \textsc{Bin Packing} problem. Suppose we are given
an instance of \textsc{Bin Packing} with $n$ items whose sizes are
specified by given values $a_{1},...,a_{n}\in\N$. Also, we are given
a bin size $B\in\N$ and a value $k\in\N$. The instance is a yes-instance
if and only if it is possible to assign the given items into at most
$k$ bins with capacity $B$ each.

For each $i\in[n]$ we introduce a job $j_{i}$ with $p_{j_{i}}=1$,
$r_{j_{i}}=0$, $d_{j_{i}}=k$, and $\H_{t}^{(j_{i})}=\{(a_{i})\}$
for each $t\in\{0,1,...,k-1\}$. We define the number of machines
by $M:=B$.

If the given instance of \textsc{Bin Packing} is a yes-instance, then
there exists a bin $b(i)\in\{0,...,k-1\}$ for each item $i\in[n]$
such that for each bin $\ell\in\{0,...,k-1\}$ the total size of the
items assigned to bin $\ell$ is bounded by $B$. We can construct
a solution for our instance of scheduling with shapes as follows.
For each $i\in[n]$ we set $S_{j_{i}}:=b(i)$. Then, for each $\ell\in\{0,...,k-1\}$
we have that during $[\ell,\ell+1)$ at most $M=B$ machines are busy
since the total size of the items in bin $\ell$ is bounded by $B$.

Conversely, suppose that there is a feasible schedule for our instance
of scheduling with shapes. For each $i\in[n]$ we assign the item
$i$ into the bin $S_{j_{i}}\in\{0,...,k-1\}$. For each $\ell\in\{0,...,k-1\}$
we have that during $[\ell,\ell+1)$ at most $M=B$ machines are busy.
Therefore, the total size of the items assigned to bin $\ell$ is
at most $M=B$ as required.
\end{proof}
On the other hand, we show that the problem is strongly NP-hard if we lift only the assumption that \mbox{$P=1$}. More precisely, we prove that this is already the case if all jobs have the same release times, deadlines, processing times, and sets of shapes, and if each job $j\in J$ must start at its release time (due to its processing time and deadline).

\begin{theorem}\label{thm:hardness-2}
	The scheduling with shapes problem is strongly NP-hard, even if $d_{j}-r_{j}=p_{j}$
	for each job $j\in J$, $\H_{S_{j}}^{(j)}=\H_{S_{j'}}^{(j')}$, $r_j= r_{j'}$, and $d_j= d_{j'}$ for
	any two jobs $j,j'\in J$, and $M=1$.
\end{theorem}

\begin{proof}
We give a reduction from the \textsc{Independent Set }problem. Suppose
we are given an undirected graph $G=(V,E)$ and an integer $k$. We
assume w.l.o.g. that $G$ does not have isolated vertices. The given
instance of \textsc{Independent Set }is a yes-instance if and only
if there exists an independent set $V'\subseteq V$ in $G$ with $|V'|=k$.

Let $n:=|V|$ and $m:=|E|$ and assume that $E=\{e_{1},...,e_{m}\}$.
We construct an instance of scheduling with shapes as follows. We
introduce $n$ jobs $J$ such that $r_{j}:=0$, $d_{j}:=n+m-k$, and
$p_{j}:=n+m-k$ for each job $j\in J$. Note that hence for each job
$j\in J$ we have that $S_{j}=0$ is the only possible start time.
We define the number of machines by $M:=1$.

For each job $j\in J$ we define $\H_{0}^{(j)}:=\H$ for a set of
shapes $\H\subseteq\{0,1\}^{n+m-k}$ defined as follows. Intuitively,
for each shape $f\in\H$ the first $m$ entries of $f$ correspond
to the $m$ edges in $E$. For each vertex $v\in V$ there is a shape
$f^{(v)}\in\H$ such that
\begin{itemize}
    \item $f_{i}^{(v)}=1$ if $i\in\{0,...,m-1\}$ and $e_{i}$ is incident
    to $v$,
    \item $f_{i}^{(v)}=0$ if $i\in\{0,...,m-1\}$ and $e_{i}$ is not incident
    to $v$, and
    \item $f_{i}^{(v)}=0$ if $i\in\{m,...,n+m-k\}$.
\end{itemize}
Also, there are $n-k$ dummy shapes $f^{\left\langle 1\right\rangle },...,f^{\left\langle n-k\right\rangle }$
such that for each dummy shape $f^{\left\langle i\right\rangle }$
we have that
\begin{itemize}
    \item $f_{m+i}^{\left\langle i\right\rangle }=1$ and
    \item $f_{i'}^{\left\langle i\right\rangle }=0$ for each $i'\in\{0,...,n+m-k\}\setminus\{m+i\}$.
\end{itemize}
We want to show that there is an independent set of size $k$ in $G$
if and only if our instance of scheduling with shapes admits a feasible
solution.

First assume that there is an independent set $V'\subseteq V$ in
$G$ of size $k$. For each job $j\in J$ we define $S_{j}:=0$ (recall
that this is the only option). For each vertex $v\in V'$ we assign
the shape $f^{(v)}$ to one (arbitrary) job $j\in J$. Thus, there
are $n-k$ jobs $J'\subseteq J$ to which we have not yet assigned
a shape. Therefore, for each $i\in[n-k]$ we assign the dummy shape
$f^{\left\langle i\right\rangle }$ to one job in $J'$. We claim
that for each $\ell\in\{0,...,n+m-k-1\}$ during $[\ell,\ell+1)$
at most one machine is busy. Assume first that $\ell\in\{0,...,m-1\}$.
Then $f_{\ell}^{\left\langle i\right\rangle }=0$ for each $i\in[n-k]$.
Moreover, $V'$ forms an independent set and, hence, there is at most
once vertex $v\in V'$ which is incident to $e_{\ell}\in E$. Therefore,
there is at most one shape $f^{(v)}\in\H$ with $f_{\ell}^{(v)}=1$
that we assigned to a job in $J$. Next, assume that $\ell\in\{m,...,m+n-k-1\}$.
Then, $f_{\ell}^{(v)}=0$ for each $v\in V$. Also, there is only
one dummy shape $f^{\left\langle i\right\rangle }$ for which $f_{\ell}^{\left\langle i\right\rangle }=1$
which the dummy shape $f^{\left\langle i\right\rangle }$ with $i=\ell$.
In particular, at most one such shape is assigned to a job. Therefore,
during $[\ell,\ell+1)$ at most one machine is busy. Thus, there is
a feasible schedule for our instance of scheduling with shapes.

Conversely, assume that there is a feasible schedule for our instance
of scheduling with shapes. For each dummy shape $f^{\left\langle i\right\rangle }$
with $i\in[n-k]$ we have that $f_{m+i}^{\left\langle i\right\rangle }=1$.
Therefore, each dummy shape can be assigned to at most one job $j\in J$.
Also, for each vertex $v\in V$ for the shape $f^{(v)}$ there is
at least one edge $e_{i}$ such that $f_{i}^{(v)}=1$ since we assumed
that $G$ does not have any isolated vertices. Therefore, the shape
$f^{(v)}$ can be assigned to at most one job $j\in J$. Let $V'$
be the set of vertices $v\in V$ for which the shape $f^{(v)}$ is
assigned to some job $j\in J$. Since each dummy shape can be assigned
to at most one job, we have that $|V'|=n-(n-k)=k$. We claim that
$V'$ is an independent set. Suppose that there are two vertices $v,v'\in V'$
which are connected by an edge. Then, there is a value $i\in[m]$
with $e_{i}=\{v,v'\}$. However, then $f_{i}^{(v)}=f_{i}^{(v')}=1$
but $M=1$ which is a contradiction. Thus, $V'$ is an independent
set of size $k$.
\end{proof}

%% file: figures/shape_demonstration.tex
\begin{figure}[ht]
\centering

\begin{tikzpicture}

% Time axis
\draw[->] (-0.5,0) -- (2.5,0) node[right] {};
\draw[->] (3.5,0) -- (6.5,0) node[right] {};
\foreach \x/\label in {0/$t$, 1/$t+1$, 2/$t+2$}
  \draw (\x,0) -- (\x,-0.1) node[below] {\label};
\foreach \x/\label in {3/$t$, 4/$t+1$, 5/$t+2$}
  \draw (\x+1,0) -- (\x+1,-0.1) node[below] {\label};

% Adjusted vertical positions for machines
\def\machiney#1{#1*0.75} % Reduce vertical spacing

% Dashed lines for each machine
% Machine r
\draw[dashed] (-0.5,\machiney{3}) -- (2.5,\machiney{3});
\draw[dashed] (-0.5,\machiney{2.25}) -- (2.5,\machiney{2.25});
\draw[dashed] (3.5,\machiney{3}) -- (6.5,\machiney{3});
\draw[dashed] (3.5,\machiney{2.25}) -- (6.5,\machiney{2.25});
\node at (-1,\machiney{2.625}) {$\mathcal{M}_r$};

% Machine l
\draw[dashed] (-0.5,\machiney{2}) -- (2.5,\machiney{2});
\draw[dashed] (-0.5,\machiney{1.25}) -- (2.5,\machiney{1.25});
\draw[dashed] (3.5,\machiney{2}) -- (6.5,\machiney{2});
\draw[dashed] (3.5,\machiney{1.25}) -- (6.5,\machiney{1.25});
\node at (-1,\machiney{1.625}) {$\mathcal{M}_\ell$};

% Machine k
\draw[dashed] (-0.5,\machiney{1}) -- (2.5,\machiney{1});
\draw[dashed] (-0.5,\machiney{0.25}) -- (2.5,\machiney{0.25});
\draw[dashed] (3.5,\machiney{1}) -- (6.5,\machiney{1});
\draw[dashed] (3.5,\machiney{0.25}) -- (6.5,\machiney{0.25});
\node at (-1,\machiney{0.625}) {$\mathcal{M}_k$};

% Add shading to the rectangles
\foreach \x/\y/\xend/\yend in 
  {4/\machiney{2.25}/5/\machiney{3}, 0/\machiney{1.25}/1/\machiney{2}, 
   5/\machiney{1.25}/6/\machiney{2}, 0/\machiney{0.25}/1/\machiney{1}, 
   1/\machiney{0.25}/2/\machiney{1}, 4/\machiney{0.25}/5/\machiney{1}}
  \draw[red,line width=0.4mm,pattern=crosshatch, pattern color=red] (\x,\y) rectangle (\xend,\yend);

\end{tikzpicture}

\caption{Two schedule options for a job at time $t$ by shape $\shape=(2,1)$ with processing time $p=2$.}\label{fig:shape}

\end{figure}

%% file: figures/voting_range.tex
\begin{figure}
\begin{center}
    \begin{tikzpicture}
        \draw[gray] (-5,0) -- (5.5,0);
            
        \draw[<->,thick, orange, line width=2pt] (-1.6,0) -- (1.5,0);
        \node[anchor=south, orange] at (0,-0.7) {$P_j$}; 
    
        % Draw five blue spots as candidates
        \foreach \x/\label in {-4.5/c_1, -2.1/c_2, -1.3/c_3, 0.9/c_4, 5.3/c_5} {
            \fill[blue] (\x,0) circle (3pt); % Draw blue circles
            \node[anchor=south] at (\x+0.05,-0.6) {$\label$};% Label above
         }   
        \node at (-1.55,0.3) {$\ell_j$};
        \node at (-2.1,-1) {$i_L=2$};
        \node at (1.45,0.3) {$u_j$};
        \node at (0.9,-1) {$i_R=4$};
    
    \end{tikzpicture}
\end{center}
\caption{Example of $i_L$ and $i_R$ for a voter $v_j$ described by $P_j$.}\label{fig:voting_range}
\end{figure}

%% file: figures/shape_creation.tex
\begin{figure}
\centering
\begin{tikzpicture}
    % Adjust vertical spacing
    \def\linegap{1.8} % Set the vertical gap between lines

    % Draw the first line
    \draw[gray] (0,0) -- (5,0);
    
    % Draw the points and label them
    \filldraw[blue] (0.6,0) circle (2pt) node[above] {\(c_1\)};
    \filldraw[blue] (3.2,0) circle (2pt) node[above] {\(c_2\)};
    \filldraw[blue] (4.5,0) circle (2pt) node[above] {\(c_3\)};
    \filldraw[red] (4.1,0) circle (2pt) node[above] {\(T'_j\)};

    % Draw the second line
    \draw[gray] (0,\linegap) -- (5,\linegap);
    
    % Draw the points and label them
    \filldraw[blue] (0.6,\linegap) circle (2pt) node[above] {\(c_1\)};
    \filldraw[blue] (3.2,\linegap) circle (2pt) node[above] {\(c_2\)};
    \filldraw[blue] (4.5,\linegap) circle (2pt) node[above] {\(c_3\)};
    \filldraw[red] (1.3,\linegap) circle (2pt) node[above] {\(T_j\)};
   
    % Arrow and L-shape for first line
    \node at (5.7,\linegap+0.3) {\(\Rightarrow\)};
    \draw[thick] (6.5,\linegap) rectangle (7,\linegap+0.5); % Bottom square
    \draw[thick] (6.5,\linegap+0.5) rectangle (7,\linegap+1.0); % Top square
    \draw[thick] (7,\linegap) rectangle (7.5,\linegap+0.5); % Right square
    
    % Arrow and L-shape for second line
    \node at (5.7,0.3) {\(\Rightarrow\)};
    \draw[thick] (6.5,0) rectangle (7,0.5); % Bottom square
    \draw[thick] (7,0.5) rectangle (7.5,1.0); % Top square
    \draw[thick] (7,0) rectangle (7.5,0.5); % Right square

    \draw[->,thick] (6.25,0) -- (7.75,0);
    \foreach \x/\label in {6.5/$2$, 7/$3$, 7.5/$4$}
        \draw (\x,0) -- (\x,-0.1) node[below] {\label};

    \draw[->,thick] (6.25,\linegap) -- (7.75,\linegap);
    \foreach \x/\label in {6.5/$1$, 7/$2$, 7.5/$3$}
        \draw (\x,\linegap) -- (\x,\linegap-0.1) node[below] {\label};
\end{tikzpicture}
    \caption{Two positions of a voter $v_j$ and the corresponding shapes.}
    \label{fig:shape_creation}
\end{figure}

%% file: figures/segments_split.tex
\begin{figure}
    \centering
    \begin{tikzpicture}
        % Horizontal axis
        \draw[gray] (-5,0) -- (9,0);

        % Points on the axis (colored blue)
        \foreach \x/\label in {-4/$c_1$, -2/$c_2$, 4.5/$c_3$, 8/$c_4$} {
            \fill[blue] (\x,0) circle (3pt) node[below] {\label};
        }

        \draw[dashed] (-3,-0.5) -- (-3,1) node[above] {$m_{1,2}$};
        \draw[dashed] (0.25,-0.5) -- (0.25,1) node[above] {$m_{1,3}$};
        \draw[dashed] (2,-0.5) -- (2,1) node[above] {$m_{1,4}$};
        \draw[dashed] (1.25,-0.5) -- (1.25,1) node[above] {$m_{2,3}$};
        \draw[dashed] (3,-0.5) -- (3,1) node[above] {$m_{2,4}$};
        \draw[dashed] (6.25,-0.5) -- (6.25,1) node[above] {$m_{3,4}$};

        % segments labels
        \node at (-4,0.8) {$E_1$};
        \node at (-1.4,0.8) {$E_2$};
        \node at (0.75,0.8) {$E_3$};
        \node at (1.6,0.8) {$E_4$};
        \node at (2.55,0.8) {$E_5$};
        \node at (4.6,0.8) {$E_6$};
        \node at (8,0.8) {$E_7$};

        % z_E labels below
        \node at (-4,-1) {$z_{E_1} = 1$};
        \node at (-1.4,-1) {$z_{E_2} = 1$};
        \node at (0.75,-1) {$z_{E_3} = 1$};
        \node at (2.55,-1) {$z_{E_5} = 2$};
        \node at (4.6,-1) {$z_{E_6} = 2$};
        \node at (8,-1) {$z_{E_7} = 2$};

        \node at (-5.3, -1.5) {$R=$};
        \node at (-4,-1.5) {$(c_1,c_2,c_3,c_4)$};
        \node at (-1.5,-1.5) {$(c_2,c_1,c_3,c_4)$};
        \node at (0.67,-1.5) {$(c_2,c_3,c_1,c_4)$};
        \node at (2.63,-1.5) {$(c_3,c_2,c_4,c_1)$};
        \node at (4.7,-1.5) {$(c_3,c_4,c_2,c_1)$};
        \node at (7.75,-1.5) {$(c_4,c_3,c_2,c_1)$};
        
        % First shape
        \draw[thick] (-4.5,-2.75) rectangle (-4.25,-2.5); 
        \draw[thick] (-4.5,-2.5) rectangle (-4.25,-2.25);
        \draw[thick] (-4.5,-2.25) rectangle (-4.25,-2);
        \draw[thick] (-4.25,-2.75) rectangle (-4,-2.5);
        \draw[thick] (-4.25,-2.5) rectangle (-4,-2.25);
        \draw[thick] (-4,-2.75) rectangle (-3.75,-2.5);
        
        % Second shape
        \draw[thick] (-1.75,-2.75) rectangle (-1.5,-2.5); 
        \draw[thick] (-1.75,-2.5) rectangle (-1.5,-2.25);
        \draw[thick] (-1.5,-2.25) rectangle (-1.25,-2);
        \draw[thick] (-1.5,-2.75) rectangle (-1.25,-2.5);
        \draw[thick] (-1.5,-2.5) rectangle (-1.25,-2.25);
        \draw[thick] (-1.25,-2.75) rectangle (-1,-2.5);
        
        % Third shape
        \draw[thick] (0.5,-2.75) rectangle (0.25,-2.5); 
        \draw[thick] (1,-2.5) rectangle (0.75,-2.25);
        \draw[thick] (0.75,-2.25) rectangle (0.5,-2);
        \draw[thick] (0.75,-2.75) rectangle (0.5,-2.5);
        \draw[thick] (0.75,-2.5) rectangle (0.5,-2.25);
        \draw[thick] (1,-2.75) rectangle (0.75,-2.5);

        % Fourth shape (shifted 0.15 to the right)
        \draw[thick] (2.4,-2.75) rectangle (2.15,-2.5); 
        \draw[thick] (2.4,-2.5) rectangle (2.15,-2.25);
        \draw[thick] (2.65,-2.25) rectangle (2.4,-2);
        \draw[thick] (2.65,-2.75) rectangle (2.4,-2.5);
        \draw[thick] (2.65,-2.5) rectangle (2.4,-2.25);
        \draw[thick] (2.9,-2.75) rectangle (2.65,-2.5);
        
        % Fifth shape
        \draw[thick] (4.5,-2.75) rectangle (4.25,-2.5); 
        \draw[thick] (5,-2.5) rectangle (4.75,-2.25);
        \draw[thick] (4.75,-2.25) rectangle (4.5,-2);
        \draw[thick] (4.75,-2.75) rectangle (4.5,-2.5);
        \draw[thick] (4.75,-2.5) rectangle (4.5,-2.25);
        \draw[thick] (5,-2.75) rectangle (4.75,-2.5);
        
        % Last shape
        \draw[thick] (7.5,-2.75) rectangle (7.25,-2.5); 
        \draw[thick] (8,-2.5) rectangle (7.75,-2.25);
        \draw[thick] (8,-2.25) rectangle (7.75,-2);
        \draw[thick] (7.75,-2.75) rectangle (7.5,-2.5);
        \draw[thick] (7.75,-2.5) rectangle (7.5,-2.25);
        \draw[thick] (8,-2.75) rectangle (7.75,-2.5);

        \draw[<->,thick, orange, line width=2pt] (-1.4,0) -- (3,0);
        \node[anchor=south, orange] at (0.6,-0.7) {$P_j$}; % Label v_j

    \end{tikzpicture}
    \caption{Partition into segments of a set of $m=4$ candidates under $3$-truncated Borda.}
    \label{fig:segment_split}
\end{figure}

%% file: figures/schedule_example.tex
\begin{figure}[ht]
\centering
\begin{tikzpicture}

% Axis
\draw[thick] (1,0) -- (9,0);
\draw[thick] (1,4) -- (9,4);

\draw[thick] (1,0) -- (1,4);
\draw[thick] (9,0) -- (9,4);
\foreach \x in {1,..., 3}    
    \draw[dashed, gray, line width=0.1mm] (1,\x) -- (9,\x);
\foreach \x in {2,..., 8}    
    \draw[dashed, gray, line width=0.1mm] (\x,0) -- (\x,4);

% Regions with patterns
\fill[pattern=north east lines, pattern color=orange!80] (1,2) -- (2,2) -- (2,3) -- (3,3) -- (3,0) -- (1,0) -- cycle;
\fill[pattern=north east lines, pattern color=cyan] (4,0) -- (4,2) -- (5,2) -- (5,1) -- (6,1) -- (6,0) -- cycle;
\fill[pattern=north east lines, pattern color=green!70!black] (7,0) -- (7,2) -- (8,2) -- (8,4) -- (9,4) -- (9,0) -- cycle;

% Step functions
\draw[line width=0.75mm, orange!80!black] (1,2) -- (2,2) -- (2,3) -- (3,3) -- (3,0);
\draw[line width=0.75mm, cyan!80!black] (4,0) -- (4,2) -- (5,2) -- (5,0);
\draw[line width=0.75mm, cyan!80!black] (4,1) -- (6,1) -- (6,0);
\draw[line width=0.75mm, green!50!black] (7,0) -- (7,2) -- (8,2) -- (8,4);

% Magenta vertical bars
\draw[dashed, line width=0.75mm, magenta] (4,2) -- (4,4);
\draw[dashed, line width=0.75mm, magenta] (6,1) -- (6,4);

% Labels under x-axis
\node at (1, -0.3) {1};
\node at (2, -0.3) {2};
\node at (3, -0.3) {3};
\node at (4, -0.3) {4};
\node at (5, -0.3) {5};
\node at (6, -0.3) {6};
\node at (7, -0.3) {7};
\node at (8, -0.3) {8};

% Label for j'
\node at (4.6,0.4) {\Large$\boldsymbol{j'}$};

\node at (1.5,3.5) {$M_1$};
\node at (2.5,3.5) {$M_2$};
\node at (7.5,3.5) {$M_7$};
\node at (8.5,3.5) {$M_8$};

\end{tikzpicture}
    \caption{Schedule example of job $j'$ in L shape}
    \label{fig:schedule_example}
\end{figure}

%% file: figures/schedule_example_left.tex
\begin{tikzpicture}

% Axis
\draw[thick] (1,0) -- (6,0);
\draw[thick] (1,4) -- (6,4);

\draw[thick] (1,0) -- (1,4);
\draw[thick] (6,0) -- (6,4);
\foreach \x in {1,..., 3}    
    \draw[dashed, gray, line width=0.1mm] (1,\x) -- (6,\x);
\foreach \x in {2,..., 5}    
    \draw[dashed, gray, line width=0.1mm] (\x,0) -- (\x,4);

% Regions with patterns
\fill[pattern=north east lines, pattern color=orange!80] (1,2) -- (2,2) -- (2,3) -- (3,3) -- (3,0) -- (1,0) -- cycle;
\fill[pattern=north east lines, pattern color=cyan] (4,0) -- (4,3) -- (6,3) -- (6,0) -- cycle;

% Step functions
\draw[line width=0.75mm, orange!80!black] (1,2) -- (2,2) -- (2,3) -- (3,3) -- (3,0);
\draw[line width=0.75mm, cyan!80!black] (4,0) -- (4,0) -- (4,3) -- (6,3);

% Labels under x-axis
\node at (1, -0.3) {1};
\node at (2, -0.3) {2};
\node at (3, -0.3) {3};
\node at (4, -0.3) {4};
\node at (5, -0.3) {5};
\node at (6, -0.3) {6};

\node at (1.5,3.5) {$M_1$};
\node at (2.5,3.5) {$M_2$};
\node at (4.5,3.5) {$M_4$};
\node at (5.5,3.5) {$M_5$};

\end{tikzpicture}

%% file: figures/schedule_example_right.tex
\begin{tikzpicture}

% Axis
\draw[thick] (1,0) -- (6,0);
\draw[thick] (1,4) -- (6,4);

\draw[thick] (1,0) -- (1,4);
\draw[thick] (6,0) -- (6,4);
\foreach \x in {1,..., 3}    
    \draw[dashed, gray, line width=0.1mm] (1,\x) -- (6,\x);
\foreach \x in {2,..., 5}    
    \draw[dashed, gray, line width=0.1mm] (\x,0) -- (\x,4);

% Regions with patterns
\fill[pattern=north east lines, pattern color=cyan] (1,3) -- (2,3) -- (2,2) -- (3,2) -- (3,0) -- (1,0) -- cycle;
\fill[pattern=north east lines, pattern color=green!50!black] (4,0) -- (4,2) -- (5,2) -- (5,4) -- (6,4) -- (6,0) -- cycle;
\draw[line width=0.75mm, green!50!black] (4,0) -- (4,2) -- (5,2) -- (5,4);

% Step functions
\draw[line width=0.75mm, cyan!80!black] (1,3) -- (2,3) -- (2,2) -- (3,2) -- (3,0);

% Labels under x-axis
\node at (1, -0.3) {4};
\node at (2, -0.3) {5};
\node at (3, -0.3) {6};
\node at (4, -0.3) {7};
\node at (5, -0.3) {8};

\node at (1.5,3.5) {$M_4$};
\node at (2.5,3.5) {$M_5$};
\node at (4.5,3.5) {$M_7$};
\node at (5.5,3.5) {$M_8$};

\end{tikzpicture}

%% file: fpt.tex
\section{Parameterized Algorithm for $\pwpar{d}$ }\label{sec:fpt}

We present a parameterized algorithm for the $\pw$ problem
in the $d$-dimensional euclidean space for any $d\ge1$. Our fixed parameter is the number of candidates $m$.

First, we describe our algorithm for positional scoring rules. Recall
that we are given a score vector $\vec{s}_m=(s_m(1),...,s_{m}(m))$ and each voter
gives a certain number of votes to each candidate, according to $\vec{s}_m$.
We say that a vector
$z=(z_{1},...,z_{m})\in\N_{0}^{m}$ is a \emph{voting
vector} if $z$ describes the number of votes that a voter may give
to each of the candidates, i.e., formally, if there is a permutation
$\sigma:[m]\rightarrow[m]$ such that $z_{i}=s_m(\sigma(i))$ for each
$i\in[m]$. We denote by $Z$ the set of all voting vectors. Recall
that each voter $v_{j}$ is described as a vector of intervals $P_{j}=\langle[\ell_{j,1},u_{j,1}],\dots,[\ell_{j,d},u_{j,d}]\rangle$.
In particular, each voter $v_{j}$ may vote only for a subset of the
voting vectors $Z$. We characterize the voters by the subsets of
$Z$ to which they may vote for. Therefore, for each subset of $Z$ we
introduce a corresponding \emph{type; }formally, we define the set
of types $\T$ to be all subsets of $Z$. We say that a voter $v_{j}$
is of some type $\tau\in\T$ if $v_{j}$ may vote for exactly the subsets
$\tau$ of $Z$. One key insight is that to solve the $\pw$ problem, for each voter $v_{j}$ we need to know only the type of
$v_{j}$. Also, there are only $|\T|=2^{|Z|}\le2^{m!}$ types which
is a value that depends only on $m$ but not on the number of voters
$n$. For each type $\tau\in\T$ denote by $n_{\tau}$ the number of voters
of type $\tau$. We can compute the type of each voter $v_{j}$ by checking
for each $z\in Z$ whether $v_{j}$ may vote according to $z$. We
can do this by solving a linear program that verifies whether there exists a valid position $T_j$ satisfying $d(T_j,c_i)\geq d(T_j, c_h)$ for every two candidates $c_i, c_h$ such that $c_i$ receives a higher score than $c_h$.

\begin{lemma}\label{lem:fpt_positional}
For each voter $v_{j}$ and each vector $z\in Z$ of a score vector $\vec{s}_m$, we
can check in polynomial time whether $v_{j}$ may vote according to
$z$.
\end{lemma}

\begin{proof}
Suppose we are given a voter $v_{j}$ and a vector $z\in Z$. W.l.o.g.~assume
that $z_{i}=s_m(i)$ for each $i\in\{1, \dots, m-1\}$. Let $A\subseteq \{1,\dots, m-1\}$ be the set of all indices for which the tie breaking rule favors $c_i$ over $c_{i+1}$. Hence, $v_{j}$ may vote
according to $z$ if and only if: for $i\in A$ and $d(T_{j},c_{i})\le d(T_{j},c_{i+1})$, or $i\in \{1, \dots, m-1\}\setminus A$ and $d(T_{j},c_{i}) < d(T_{j},c_{i+1})$.
The latter condition can be written in the
form $a_{i}^{\top}T_{j}\le b_{i}$ for some vector $a_{i}\in\R^{d}$
and some scalar $b_{i}\in\R$ since all points with equal distance
to $c_{i}$ and $c_{i+1}$ lie on a hyperplane in $\R^{d}$. In the case of a strong inequality we add a variable which we aim to maximize, $a_{i}^{\top}T_{j}\le b_{i}+\varepsilon$. If $\varepsilon > 0$ then $a_{i}^{\top}T_{j} < b_{i}$.
We notice that in the case where $s_m(i)=s_m(i+1)$ we simply omit the inequality, because the order between the two candidates is irrelevant. 
Thus, $v_{j}$ may vote according to $z$ if and only if the following linear
program has a solution with $\varepsilon>0$, which we can check in polynomial time.
\begin{alignat*}{2}
\text{maximize} & \quad \varepsilon & & \\
a_{i}^{\top}T_{j} & \le b_{i} & \quad & \forall i\in A\\
a_{i}^{\top}T_{j} & \le b_{i}+\varepsilon & \quad & \forall i\in [m-1]\setminus A\\
T_{j,k} & \ge\ell_{j,k} &  & \forall k\in[d]\\
T_{j,k} & \le u_{j,k} &  & \forall k\in[d]\\
T_{j,k} & \in\R &  & \forall k\in[d] \\
\varepsilon \geq 0 & & &
\end{alignat*}
\end{proof}

Let $i^{*}\in[m]$ be the index of the candidate $c^{*}$ for which
we want to determine whether it can win the election, i.e., $c_{i^{*}}=c^{*}$.
We formulate an integer linear program that tries to compute an outcome
of the election in which $c^{*}$ wins. For each type $\tau\in\T$ and each voting vector $z\in Z$ we introduce a variable $x_{\tau}^{z}$
which denotes the number of voters of type $\tau$ that vote according
to the voting vector $z$.

\begin{alignat*}{2}
\sum_{\tau\in\T}\sum_{z\in Z}x_{\tau}^{z}\cdot z_{i} & \leq M^{*} &  & \forall i=[m]\setminus\{i^{*}\}\\[1ex]
\sum_{\tau \in\T}\sum_{z\in Z}x_{\tau}^{z}\cdot z_{i^{*}} & =M^{*}\\
\sum_{z\in Z}x_{\tau}^{z} & =n_{\tau} &  & \forall \tau\in\T\\[1ex]
x_{\tau}^{z} & \in\mathbb{N}_{0} & \quad & \forall \tau\in\T,\forall z\in Z\\[1ex]
M^{*} & \in\mathbb{N}
\end{alignat*}
The integer program has a solution if and only if there is an outcome
of the election in which $c^{*}$ receives $M^{*}$ votes (for some
value $M^{*}\in\N$) and no other candidate receives more than $M^{*}$
votes, i.e., $c^{*}$ is a possible winner. The number of variables
is bounded by $1+|\T||Z|\le1+m!\cdot2^{m!}$. Hence, we can solve
the program in a running time of the form $(\log(s_m(1)))^{O(1)}f(m)$ 
using algorithms for integer programs in fixed dimensions, e.g., \cite{lenstra1983integer,ReisRothvoss2023}. A similar technique is used, e.g. in \cite{kimelfeld2019query}.

\begin{theorem}\label{theorem:approval_positional}For every positional
scoring rule and any $d\ge1$, $\pwpar{d}$ can be solved
in time $(n\cdot\log(s_m(1)))^{O(1)}f(m)$ for some function
$f$, i.e., $\pwpar{d}$ is FPT for the parameter $m$.
\end{theorem}

Our algorithm can be adjusted to the setting of approval voting:
we set $Z:=\{0,1\}^{m}$, i.e., all combinations of partitioning
the candidates into approved and unapproved candidates.
Then, for a voter $v_j$ and a voting vector $z\in Z$, $v_j$ can vote by $z$ if there is a valid position $T_j$ such that for every $i\in [m]$, if $z_i=1$ then $d(T_j,c_i)\leq \rho_j$, and if $z_i=0$, $d(T_j,c_i)> \rho_j$. This can be checked by solving a set of inequalities, which by Grigor'ev and Vorobjov~\cite{grigor1988solving} can be solved in $O(f(m))$ time.

\begin{lemma}\label{lem:fpt_approval}
For each voter $v_{j}$ and each vector $z\in Z$ of approval voting, we
can check in polynomial time whether $v_{j}$ may vote according to
$z$.
\end{lemma}

\begin{proof}
Suppose we are given a voter $v_{j}$ and a vector $z\in Z$. $v_{j}$ may vote
according to $z$ if and only if there exists a valid position $T_j$ such that for every $i\in [m]$, if $z_i=1$ then $d(T_j,c_i)\leq \rho_j$, and if $z_i=0$, $d(T_j,c_i)> \rho_j$. This condition creates a systems of $m+d$ inequalities with $d$ variables and a maximal degree of 2.
\begin{alignat*}{2}
d(T_j,c_i)>\rho_j & \quad &\forall i \in[m]: z_i=0\\
\rho_j- d(T_j,c_i)\geq 0 & \quad & \forall i\in[m]: z_i=1\\
T_{j,\ell}\geq l_{j,\ell} & \quad & \forall \ell\in [d]\\
T_{j,\ell}\leq u_{j,\ell} & \quad & \forall \ell\in [d]\\
\end{alignat*}
By Grigor'ev and Vorobjov~\cite{grigor1988solving} a solution for this system of inequations can be found in time polynomial in $(m\cdot 2)^{d^2}$.
\end{proof}

\begin{theorem}\label{theorem:approval_fpt} For any fixed $d\geq1$, $\pwpar{d}$ with approval voting
can be solved in time $n^{O(1)}f(m)$ for some function $f$, i.e.,
it is FPT for the parameter~$m$. \end{theorem}

%% file: weighted.tex
In weighted spatial voting, every voter $v_j$ is associated with a weight $w_j$, and the score contributed by voter \( v_j \) to candidate \( c \) is 
\( s(R_j, c) = w_j \cdot s_m(i) \), where \( c \) is ranked in position \( i \) according to \( v_j \)’s preference \( R_j \), and \( (s_m(1), \dots, s_m(m)) \) represents the score vector.

The $\nw$ problem in weighted spatial voting remains traceable for every positional scoring rule and fixed dimension, using the algorithm in~\cite{imber2024spatial} for the unweighted variant. Indeed, we can solve the problem by computing the maximal score difference $s(R_j ,c)-s(R_j ,c^*)$ across all ranking completions $R_j$ of $P_j$ for every candidate $c\neq c^*$, as in the unweighted case.

We investigate the $\pw$ problem in the weighted spatial voting model in one dimension, denoted as $\wpwpar{1}$. We start with two-valued positional scoring rules, which we denote as $k(m)$-approval, and distinguish between rules that are traceable and rules that are NP-complete.

\begin{theorem}\label{theorem:k(m)_npc}
    Let $k(m)$-approval be a two-valued scoring rule. If for every $m\in\N$, it holds that $k(m) \geq \frac{m}{2}$, $\wpwpar{1}$ with $k(m)$-approval is in P. Otherwise, it is NP-complete.
\end{theorem}

\begin{proof}
Let $k(m)$ be a function such that $k(m) \geq \frac{m}{2}$ for all $m\in \N$.
Given an instance with $m$ candidates, let $k=k(m)$. We prove separately for $k = \frac{m}{2}$ and $k > \frac{m}{2}$. When $k > \frac{m}{2}$, candidates $c_{m-k+1}, \dots, c_k$ are always in the top $k$, receiving maximal scores. Any $c^*$ in this set is a possible winner. A candidate not in this set can only be a possible winner if there exists a profile completion placing $c^*$ in the top $k$ of every voter. This can be verified in polynomial time by segmenting the space by midpoints involving $c^*$, determining the top $k$ candidates for each segment, and verifying $c^*$'s position for all voters.

For $k = \frac{m}{2}$, w.l.o.g $c^*$ is in the first half of the candidates. We prove $c^*$ is a possible winner if and only if it is a possible winner under a specific profile completion $\mathbf{T}$, in which every voter $v_j$ that can vote for $c^*$ is positioned at $T_j=\ell_j$, and the rest are positioned at $T_j=u_j$. For the forward direction, starting from a profile completion $\mathbf{T}'$ where $c^*$ is a winner, we adjust voters one by one.
\begin{enumerate}
    \item 
    If $v_j$ can vote for $c^*$, then by moving its position to $\ell_j$ the scores for candidates $c > c^*$ increase by $w_j$ only if $c^*$'s score also increases. 
    Candidates $c < c^*$ never outscore $c^*$ since voters for $c$ also vote for $c^*$.
    \item 
    If $v_j$ cannot vote for $c^*$, then it must vote for $c_{\frac{m}{2}+1}$. By moving the position to $u_j$ only candidates $c > c_{\frac{m}{2}+1}$ may increase their scores, but such candidates cannot outscore $c_{\frac{m}{2}+1}$, which remains with the same score as before, therefore does not surpass $c^*$. 
\end{enumerate}
In both cases, $c^*$ remains a possible winner. After all adjustments, $c^*$ is a possible winner under $\mathbf{T}$. 

We now discuss the case where there exists $m\in \N$ such that $k(m)<\frac{m}{2}$. It is clear that this problem is in NP by guessing a voting profile, calculating the score of each candidate and accepting the instance if no other candidate $c\neq c^*$ receives a higher score than $c^*$.

Let $m\in N$ be a value for which $k(m)<\frac{m}{2}$, and let $k=k(m)$ for that $m$.
We prove NP-hardness separately for the case where $k=1$, which creates the plurality voting rule, and the case where $2\leq k$. In both cases we give a reduction from \textsc{partition}. We begin with the case were $k=1$. Let $\{a_1, \dots, a_n\}$ be a set of $n$ distinct positive integers that sum to $2A$, we form the following instance of $PW\langle 1 \rangle$. Let $C=\{c_1, c_{2}, c^*\}$ be the set of candidates and their position on the axis are $c_1= 1, c_2 = 2$, $c^*= 4$. We define $n+1$ voters. For every voter $v_i$ when $i\in [n]$, we define its partial profile to be $P_j=[1,2]$ and its weight $w_i=a_i$. We add an additional voter: $v_{n+1}$ with $P_{n+1}=[4,5]$ and $w_{n+1}=A$. Note that the set of candidates can be enlarged to any size by adding candidates that are positioned on the axis far enough such that they would not be in the top preference for any voter.
We prove that $c^*$ is a possible winner if and only if $\{a_1, \dots, a_n\}$ can be partitioned into two subsets that sum to $A$.
Note that for every position $T_{n+1} \in [4,5]$ of $v_{n+1}$ results in $R_{n+1}=(c^*, c_2, c_1)$, therefore $s(R_{n+1},c^*) = A$.

Assuming $\{a_1, \dots, a_n\}$ can be partitioned into two subsets that sum to $A$, denoted $S_1$ and $S_2$. We construct a spatial completion $\mathbf{T}$ in the following way. For every $i:\; a_i\in S_1$ we set $T_i=1$, and for every $i:\; a_i\in S_2$ we set $T_i=2$. Then:
\begin{itemize}
    \item $s(\mathbf{R_T},c_1) = \sum_{a_i\in S_1} s(R_{T_i},c_1) = \sum_{a_i\in S_1} w_i = A $
    \item $s(\mathbf{R_T},c_2) = \sum_{a_i\in S_2} s(R_{T_i},c_2) = \sum_{a_i\in S_2} w_i = A $
    \item $s(\mathbf{R_T},c^*) = s(R_{n+1},c^*) = A $
\end{itemize}
Making $c^*$ a possible winner. 

We continue with the other direction, in this case we assume that $c^*$ is a possible winner, meaning there is a spatial completion $\mathbf{T}$ such that $s(\mathbf{R_T},c^*)\geq s(\mathbf{R_T},c_i)$ for all $i\in\{1,2\}$. As explained before, $T_{n+1}$ will always result in the same ranking profile. 

Because $c^*$ can receive no more than $A$ votes, every other candidate must receive at most $A$ votes from the rest of the voters. We look at the rest of the voters, which can be positioned in $[1, 2]$. W.l.o.g the tie breaking in case of equal distance is in favor of $c_1$. A voter $v_i$ with position $T_i\in[1, 1.5]$, would cast a score of $w_i=a_i$ to candidate $c_1$. The rest of the voters, with positions $T_i\in (1.5,2]$ cast their votes to $c_2$. Let $S_1$ be the set of all voters with position $T_i\in [1,1.5]$, and $S_2$ the rest of them. Then:

$$A\geq s(\mathbf{R_T},c_1) = \sum_{v_i\in S_1} s(R_{T_i},c_1) = \sum_{v_i\in S_1} w_i = \sum_{v_i\in S_1} a_i$$

$$A\geq s(\mathbf{R_T},c_2) = \sum_{v_i\in S_2} s(R_{T_i},c_2) = \sum_{v_i\in S_2} w_i = \sum_{v_i\in S_2} a_i$$

Because all voters besides $v_{n+1}$ are at either $S_1$ or $S_2$, $\sum_{v_i\in S_1} a_i+\sum_{v_i\in S_2} a_i=2A$, meaning the sum of each group is exactly $A$, and $S_1,S_2$ are the wanted partition. This concludes the proof for $k=1$.

We continue with the case of $2\leq k $. Let $\{a_1, \dots, a_n\}$ be a set of $n$ distinct positive integers that sum to $2A$, we form the following instance of $PW\langle 1 \rangle$. Let $C=\{c_1, \dots, c_{2k}, c^*\}$ be the set of candidates and their position on the axis are for all $i\leq k$, $c_i= i-1$, $c^*= k$, and for all $i\geq k+1, c_i=i$. We define $n+2$ voters. For every voter $v_i$ when $i\in [n]$, we define $P_i=[\frac{k+1}{2},\frac{3k}{2}]$ and its weight $w_i=a_i$. We add 2 more voters: $v_{n+1}$ with $P_{n+1}=[\ell_{n+1},u_{n+1}]=[-1,0]$ and $w_{n+1}=A$, and $v_{n+2}$ with $P+{n+2}=[\ell_{n+2},u_{n+2}]=[2k,2k+1]$ and $w_{n+1}=A$.
We prove that $c^*$ is a possible winner if and only if $\{a_1, \dots, a_n\}$ can be partitioned into two subsets that sum to $A$.

Note that for every position $T_{n+1}\in[-1,0]$ of $v_{n+1}$ results in $R_{n+1}=(c_1, \dots, c_k, c^*, c_{k+1}, \dots, c_{2k})$, and every position $T_{n+2}\in[2k, 2k+1]$ of $v_{n+2}$ results in $R_{n+2} =(c_{2k}, \dots, c_{k+1}, c^*, c_{k}, \dots, c_{1})$, concluding that in every spatial completion both voters contributes a score of $A$ to every candidate except $c^*$.

Assuming $\{a_1, \dots, a_n\}$ can be partitioned into two subsets that sum to $A$, denoted $S_1$ and $S_2$. We construct a spatial completion $\mathbf{T}$ in the following way. For every $i:\; a_i\in S_1$ we set $T_i=\frac{k+1.5}{2}$, and for every $i:\; a_i\in S_2$ we set $T_i=\frac{3k-1}{2}$. These positions creates a ranking profile in which $c_2, \dots c_k, c^*$ are the top $k$ for every voter $v_i$ such that $a_i\in S_1$, and a ranking profile in which $c^*, c_{k+1}, \dots , c_{2k-1}$ for every voter $v_i$ such that $a_i\in S_2$. Combined with the votes of $v_{n+1}$ and $v_{n+2}$, the final scores of candidates $c_2, \dots, c_r, c_{k+1}, \dots, c_{2k-1}$ is $2A$ while for $c_1$ and $c_{2k}$ the final score is $A$, making $c^*$ a possible winner.

We continue with the other direction, in this case we assume that $c^*$ is a possible winner, meaning there is a spatial completion $\mathbf{T}$ such that $s(\mathbf{R_T},c^*)\geq s(\mathbf{R_T},c_i)$ for all $i$. As explained before, $T_{n+1}$ and $T_{n+2}$ will always result in the same ranking profile. 

Because $c^*$ can receive no more than $2A$ votes, every other candidate must receive at most $A$ additional votes from the rest of the voters.
We look at the rest of the voters, which can be positioned in $[0, 2k]$:
\begin{enumerate}
    \item For voters $v_i$ such that $T_i\in [\frac{k+1}{2},\frac{k+2}{2}]$, the ranking profile $R_{T_i}$ would have candidates $c_2, \dots, c_k, c^*$ in the top $k$ ranks. Because $m_{c_1,c^*} = \frac{k}{2})$ and $m_{c_2,c_{k+1}}=\frac{1+k+1}{2}$.    
    \item For voters $v_i$ such that $T_i\in (\frac{3k-2}{2},\frac{3k}{2}]$, the ranking profile $R_{T_i}$ would have candidates $c^*, c_{k+1}, \dots, c_{2k-1}$ in the top $k$ ranks. Because $m_{c_k,c_{2k-2}}=\frac{3k-2}{2}$ and $m_{c^*,c_{2k}}=\frac{k+2k}{2}$).
    \item For voters in the remaining section, for which $T_i\in (\frac{k+2}{2}, \frac{3k-2}{2}]$, both $c_k, c^*$ and $c_{k+1}$ are in the top $k$ places in $R_{T_i}$.
\end{enumerate}

Let $z_1,z_2,z_3$ be the sum of weights of voters in each group. Note that $z_1+z_2+z_3=2A$ as each voter must be in one of these groups. Then, by the $k$-approval rule, the score sum of each candidate is:
\begin{itemize}
    \item $s(\mathbf{R_T},c_k) = A+z_1+z_3$
    \item $s(\mathbf{R_T},c^*) = z_1+z_2+z_3$
    \item $s(\mathbf{R_T},c_{k+1}) = A+z_2+z_3$
\end{itemize}
Candidate $c^*$ is a possible winner only if
$s(\mathbf{R_T},c^*)\geq s(\mathbf{R_T},c_k)$:
$$z_1+z_2+z_3\geq A+z_1+z_3 \quad \Rightarrow \quad z_2\geq A$$
And only if $s(\mathbf{R_T},c^*)\geq s(\mathbf{R_T},c_{k+1})$:
$$z_1+z_2+z_3\geq A+z_2+z_3 \quad \Rightarrow \quad z_1 \geq A$$

Because $z_1+z_2+z_3=2A$, it must be that $z_1=z_2=A$ and $z_3=0$, meaning the total weight of the voters in each group 1 and 2 is summed up to $A$. We define the corresponding elements in each group to be a set, and because the weights of the voters is the same as the elements, this induces a partition.
\end{proof}

Next, we give a hardness result for the Borda voting rule. The proof idea is similar to the proof of Theorem 4.3 in~\cite{faliszewski2009shield}, which proves that for single-peaked preferences, the constructive coalition weighted manipulation problem is NP-complete.

\begin{theorem}\label{theorem:weighted_borda}
    $\wpwpar{1}$ with the Borda voting rule is NP-complete already when the number of candidates is $m=4$.
\end{theorem}

\begin{proof}
The problem is in NP by guessing a voting profile, calculating the score of every candidate, and accepting the instance if no other candidate $c\neq c^*$ receives a higher score than $c^*$.

We present a reduction from \textsc{partition}. Let $\{a_1, \dots, a_n\}$ be a set of $n$ distinct positive integers that sum to $2A$, we form the following instance of $\pwpar{1}$. Let $C=\{c_1, c_2, c_3, c^*\}$ be the set of candidates and their position on the axis are $c_1= 0, c_2= 1, c^*= 2$ and $c_3=5$. For simplicity, assume distance ties are broken in favor of the candidate to the right. We define $2+n$ voters. For every voter $v_i$ when $i\in [n]$, we set $P_i = [2,3.5]$ and its weight $w_i=a_i$. We add 2 more voters: $v_{n+1}$ with $P_{n+1}=[5,6]$ and $w_{n+1}=11A$, and $v_{n+2}$ with $P_{n+2}=[1.6,2]$ and $w_{n+1}=7A$.
We prove that $c^*$ is a possible winner if and only if $\{a_1, \dots, a_n\}$ can be partitioned into two subsets, each sums to $A$.

Note that for every position $T_{n+1}\in [5,6]$ of $v_{n+1}$  results in $R_{n+1}=(c_3, c^*, c_2, c_1)$, and every position $T_{n+1}\in [1.6, 2]$ of $v_{n+2}$ results in $R_{n+2} =(c_2, c_1, c^*, c_3)$, concluding that in every spatial completion both voters contributes a score of $14A$ to $c_1$, $32A$ to $c_2$, $29A$ to $c^*$, and $33A$ to $c_3$.

Assuming $\{a_1, \dots, a_n\}$ can be partitioned into two subsets that sum to $A$, denoted $S_1$ and $S_2$. We construct a spatial completion $\mathbf{T}=(T_1, \dots, T_n)$ in the following way. For every $i:\; a_i\in S_1$ we set $T_i=2$, and for every $i:\; a_i\in S_2$ we set $T_i=3.5$. These positions create $R_i = (c^*, c_2, c_1, c_3)$ as the ranking profile for every voter $v_i$ such that $a_i\in S_1$ and $R_i = (c^*, c_3, c_2, c_1)$ as the ranking profile for every voter $v_i$ such that $a_i\in S_2$. Combined with the votes of $v_{n+1}$ and $v_{n+2}$, the final scores of the candidates are $S(\mathbf{R},c_1)=15A$, $S(\mathbf{R},c_2)=35A$, $S(\mathbf{R},c^*)=35A$, and $S(\mathbf{R},c_3)=35A$, making $c^*$ a possible winner.

We continue with the other direction, in this case we assume that $c^*$ is a possible winner, meaning there is a spatial completion $\mathbf{T}$ such that $s(\mathbf{R_T},c^*)\geq s(\mathbf{R_T},c_i)$ for all $i$. As explained before, $T_{n+1}$ and $T_{n+2}$ will always result in the same ranking profile. We look at the rest of the voters, which can be positioned in $[2, 3.5]$. For voters $v_i$ such that $T_i\in [2,2.5]$, the ranking profile would be $\mathbf{R_T}=(c^*,c_2,c_1, c_3)$, for voters $v_i$ such that $T_i\in (2.5,3]$, the ranking profile would be $\mathbf{R_T}=(c^*,c_2,c_3, c_1)$, and for voters $v_i$ such that $T_i\in (3,3.5]$, the ranking profile would be $\mathbf{R_T}=(c^*,c_3,c_2, c_1)$. Let $z_1,z_2$ and $z_3$ be the sum of weights of voters in each of the following groups. Note that $z_1+z_2+z_3=2A$ as each voter must be in one of these groups. Then, by the Borda rule, the score sum of each candidate is:

\begin{itemize}
    \item $s(\mathbf{R_T},c_1) = 14A+z_1$
    \item $s(\mathbf{R_T},c_2) = 32A+2z_1+2z_2+z_3$
    \item $s(\mathbf{R_T},c^*) = 29A+3z_1+3z_2+3z_3=35A$
    \item $s(\mathbf{R_T},c_3) = 33A+z_2+2z_3$
\end{itemize}

For all $i$: $s(\mathbf{R_T},c^*)\geq s(\mathbf{R_T},c_i)$ , therefore:
\begin{itemize}
    \item $35A\geq 14A+z_1$
    \item $35A\geq 32A+z_1+z_2+z_3 = 34A+2z_1+2z_2$
    \item $35A\geq 33A+z_2+2z_3$
\end{itemize}

The first equation holds since $z_1\leq 2A$.
By the second equation, 
$$A\geq z_1+z_2$$
We add $z_3$ to both sides and get that $z_3\geq A$.
By the third equation:
$$2A\geq z_2+2z_3\geq z_2+2A$$
Therefore $z_2=0$.
By the previous equation:
$$2A\geq z_2+2z_3 \quad \Rightarrow \quad A\geq z_3 \quad \Rightarrow \quad z_3=A$$   
Finally,
$$2A = z_1+z_2+z_3 \quad \Rightarrow \quad z_1=A$$

This implies that the weights of voter positioned at $T_i\in [2,2.5]$ is equal to the weight of voters positioned at $T_i\in (3,3.5]$, which is half of the total weight of voters $v_1, \dots, v_n$. We define the corresponding elements in each range to be a set, and because the weights of the voters is the same as the elements, this creates a correct partition.
\end{proof}

%% file: discussion.tex
 In this paper we investigated the computational complexity of $\pw$, which naturally arises in spatial voting with incomplete voters' information. There are several interesting directions for future work. While we show that $\pwpar{1}$ is in P for any $k$-truncated scoring rule and any constant $k$,
 the computational complexity of the problem remains open under some natural scoring rules such as Borda. We note that a hardness result for $\pwpar{1}$ under Borda would resolve also the computational complexity of manipulation under Borda in the single-peaked model, which has been open for over a decade \cite{faliszewski2009shield}. It would also be interesting to find a natural parameter for which $\wpwpar{1}$ is FPT. 
 Finally, $\wpwpar{d}$ remains open already under certain two-valued positional scoring rules when $d \geq 2$.

%% file: main.bbl
\begin{thebibliography}{10}

\bibitem{AGG2015}
C.~Alós-Ferrer and G.~D. Granić.
\newblock Political space representations with approval data.
\newblock {\em Electoral Studies}, 39:56--71, Jan. 2015.

\bibitem{baptiste2000scheduling}
P.~Baptiste.
\newblock Scheduling equal-length jobs on identical parallel machines.
\newblock {\em Discrete Applied Mathematics}, 103(1-3):21--32, 2000.

\bibitem{bartholdi1989computational}
J.~J. Bartholdi, C.~A. Tovey, and M.~A. Trick.
\newblock The computational difficulty of manipulating an election.
\newblock {\em Social choice and welfare}, 6:227--241, 1989.

\bibitem{baumeister2012campaigns}
D.~Baumeister, P.~Faliszewski, J.~Lang, and J.~Rothe.
\newblock Campaigns for lazy voters: truncated ballots.
\newblock In {\em AAMAS}, pages 577--584, 2012.

\bibitem{baumeister2011computational}
D.~Baumeister, M.~Roos, and J.~Rothe.
\newblock Computational complexity of two variants of the possible winner problem.
\newblock In {\em The 10th International Conference on Autonomous Agents and Multiagent Systems-Volume 2}, pages 853--860, 2011.

\bibitem{baumeister2012possible}
D.~Baumeister, M.~Roos, J.~Rothe, L.~Schend, and L.~Xia.
\newblock The possible winner problem with uncertain weights.
\newblock In {\em ECAI 2012}, pages 133--138. IOS Press, 2012.

\bibitem{baumeister2012taking}
D.~Baumeister and J.~Rothe.
\newblock Taking the final step to a full dichotomy of the possible winner problem in pure scoring rules.
\newblock {\em Information Processing Letters}, 112(5):186--190, 2012.

\bibitem{betzler2010towards}
N.~Betzler and B.~Dorn.
\newblock Towards a dichotomy for the possible winner problem in elections based on scoring rules.
\newblock {\em Journal of Computer and System Sciences}, 76(8):812--836, 2010.

\bibitem{black1948rationale}
D.~Black.
\newblock On the rationale of group decision-making.
\newblock {\em Journal of political economy}, 56(1):23--34, 1948.

\bibitem{bogomolnaia2007euclidean}
A.~Bogomolnaia and J.-F. Laslier.
\newblock Euclidean preferences.
\newblock {\em Journal of Mathematical Economics}, 43(2):87--98, 2007.

\bibitem{brandt2015bypassing}
F.~Brandt, M.~Brill, E.~Hemaspaandra, and L.~A. Hemaspaandra.
\newblock Bypassing combinatorial protections: Polynomial-time algorithms for single-peaked electorates.
\newblock {\em Journal of Artificial Intelligence Research}, 53:439--496, 2015.

\bibitem{brandt2016handbook}
F.~Brandt, V.~Conitzer, U.~Endriss, J.~Lang, and A.~D. Procaccia.
\newblock {\em Handbook of computational social choice}.
\newblock Cambridge University Press, 2016.

\bibitem{chakraborty2021classifying}
V.~Chakraborty and P.~G. Kolaitis.
\newblock Classifying the complexity of the possible winner problem on partial chains.
\newblock In {\em AAMAS'21: Proceedings of the 20th International Conference on Autonomous Agents and MultiAgent Systems}, 2021.

\bibitem{conitzer2002complexity}
V.~Conitzer and T.~Sandholm.
\newblock Complexity of manipulating elections with few candidates.
\newblock In {\em AAAI/IAAI}, pages 314--319, 2002.

\bibitem{conitzer2007elections}
V.~Conitzer, T.~Sandholm, and J.~Lang.
\newblock When are elections with few candidates hard to manipulate?
\newblock {\em Journal of the ACM (JACM)}, 54(3):14--es, 2007.

\bibitem{cygan2015parameterized}
M.~Cygan, F.~V. Fomin, {\L}.~Kowalik, D.~Lokshtanov, D.~Marx, M.~Pilipczuk, M.~Pilipczuk, and S.~Saurabh.
\newblock {\em Parameterized algorithms}, volume~5.
\newblock 2015.

\bibitem{dougan2014implementing}
O.~Do{\u{g}}an and A.~E. Giritligil.
\newblock Implementing the borda outcome via truncated scoring rules: a computational study.
\newblock {\em Public Choice}, 159:83--98, 2014.

\bibitem{downey2013fundamentals}
R.~G. Downey, M.~R. Fellows, et~al.
\newblock {\em Fundamentals of parameterized complexity}, volume~4.
\newblock 2013.

\bibitem{elkind2022preference}
E.~Elkind, M.~Lackner, and D.~Peters.
\newblock Preference restrictions in computational social choice: A survey.
\newblock {\em arXiv preprint arXiv:2205.09092}, 2022.

\bibitem{faliszewski2009shield}
P.~Faliszewski, E.~Hemaspaandra, L.~A. Hemaspaandra, and J.~Rothe.
\newblock The shield that never was: Societies with single-peaked preferences are more open to manipulation and control.
\newblock In {\em Proceedings of the 12th Conference on Theoretical Aspects of Rationality and Knowledge}, pages 118--127, 2009.

\bibitem{grigor1988solving}
D.~Y. Grigor'ev and N.~N. Vorobjov~Jr.
\newblock Solving systems of polynomial inequalities in subexponential time.
\newblock {\em Journal of symbolic computation}, 5(1-2):37--64, 1988.

\bibitem{imber2024spatial}
A.~Imber, J.~Israel, M.~Brill, H.~Shachnai, and B.~Kimelfeld.
\newblock Spatial voting with incomplete voter information.
\newblock In {\em Proceedings of the AAAI Conference on Artificial Intelligence}, volume~38, pages 9790--9797, 2024.

\bibitem{kenig2019complexity}
B.~Kenig.
\newblock The complexity of the possible winner problem with partitioned preferences.
\newblock In {\em Proceedings of the 18th International Conference on Autonomous Agents and MultiAgent Systems}, pages 2051--2053, 2019.

\bibitem{kimelfeld2019query}
B.~Kimelfeld, P.~G. Kolaitis, and M.~Tibi.
\newblock Query evaluation in election databases.
\newblock In {\em Proceedings of the 38th ACM SIGMOD-SIGACT-SIGAI Symposium on Principles of Database Systems}, pages 32--46, 2019.

\bibitem{konczak2005voting}
K.~Konczak and J.~Lang.
\newblock Voting procedures with incomplete preferences.
\newblock In {\em Proc. IJCAI-05 Multidisciplinary Workshop on Advances in Preference Handling}, volume~20, 2005.

\bibitem{lenstra1983integer}
H.~W. Lenstra~Jr.
\newblock Integer programming with a fixed number of variables.
\newblock {\em Mathematics of operations research}, 8(4):538--548, 1983.

\bibitem{moulin1984generalized}
H.~Moulin.
\newblock Generalized condorcet-winners for single peaked and single-plateau preferences.
\newblock {\em Social Choice and Welfare}, 1(2):127--147, 1984.

\bibitem{niedermeier2002invitation}
R.~Niedermeier.
\newblock Invitation to fixed-parameter algorithms.
\newblock {\em Habilitationschrift, University of T{\"u}bingen}, 19, 2002.

\bibitem{pini2011incompleteness}
M.~S. Pini, F.~Rossi, K.~B. Venable, and T.~Walsh.
\newblock Incompleteness and incomparability in preference aggregation: Complexity results.
\newblock {\em Artificial Intelligence}, 175(7-8):1272--1289, 2011.

\bibitem{ReisRothvoss2023}
V.~Reis and T.~Rothvoss.
\newblock The subspace flatness conjecture and faster integer programming.
\newblock In {\em 2023 IEEE 64th Annual Symposium on Foundations of Computer Science (FOCS 2023)}, pages 974--988, 2023.

\bibitem{SBKC18}
D.~Stockemer, A.~Blais, F.~Kostelka, and C.~Chhim.
\newblock Voting in the eurovision song contest.
\newblock {\em Politics}, 38(4):428--442, 2018.

\bibitem{terzopoulou2021borda}
Z.~Terzopoulou and U.~Endriss.
\newblock The borda class: An axiomatic study of the borda rule on top-truncated preferences.
\newblock {\em Journal of Mathematical Economics}, 92:31--40, 2021.

\bibitem{walsh2007uncertainty}
T.~Walsh.
\newblock Uncertainty in preference elicitation and aggregation.
\newblock In {\em AAAI}, volume~7, pages 3--8, 2007.

\bibitem{yang2017complexity}
Y.~Yang.
\newblock On the complexity of borda control in single-peaked elections.
\newblock In {\em Proceedings of the 16th Conference on Autonomous Agents and MultiAgent Systems}, pages 1178--1186, 2017.

\end{thebibliography}
